\documentclass[aps,physrev,twocolumn,superscriptaddress]{revtex4-2}

\usepackage{amsthm, amsmath, amsfonts, graphicx}
\usepackage[dvipsnames]{xcolor}

\usepackage{xcolor}
\definecolor{dark-red}{rgb}{0.4,0.15,0.15}
\definecolor{dark-blue}{rgb}{0.15,0.15,0.4}
\definecolor{medium-blue}{rgb}{0,0,0.5}
\usepackage{hyperref}
\usepackage[all]{hypcap}
\usepackage[capitalize]{cleveref}
\hypersetup{
    colorlinks, linkcolor={dark-red},
    citecolor={dark-blue}, urlcolor={medium-blue}
} 

\newtheorem{thm}{Theorem}
\newtheorem{fact}{Fact}

\newtheorem{cor}{Corollary}[section]
\newtheorem{prop}{Proposition}[section]
\newtheorem*{prop*}{Proposition}
\newtheorem{lem}{Lemma}[section]
\newtheorem*{lem*}{Lemma}
\newtheorem{defn}{Definition}[section]
\newtheorem*{defn*}{Definition}

\usepackage{hyperref}
\usepackage{cleveref}

\crefname{equation}{Eq.}{Eqs.}
\crefname{algorithm}{Alg.}{Algs.}
\crefname{appendix}{App.}{Apps.}
\crefname{figure}{Fig.}{Figs.}
\crefname{section}{Sec.}{Secs.}
\crefname{defn}{Def.}{Defs.}
\crefname{prop}{Prop.}{Propositions}
\crefname{lem}{Lem.}{Lemmas}
\crefname{thm}{Thm.}{Theorems}
\crefname{fact}{Fact.}{Facts.}

\begin{document}


\title{Learning T-conjugated stabilizers: The multiple-squares dihedral StateHSP}


\author{Gideon Lee}
\email[]{gideonlee@uchicago.edu}
\affiliation{Google Quantum AI, Venice, CA 90291, USA}
\affiliation{Pritzker School of Molecular Engineering, The University of Chicago,
Chicago, Illinois 60637, USA}

\author{Jonathan A. Gross}
\email[]{jarthurgross@google.com}
\affiliation{Google Quantum AI, Venice, CA 90291, USA}

\author{Masaya Fukami}
\affiliation{Google Quantum AI, Venice, CA 90291, USA}

\author{Zhang Jiang}
\affiliation{Google Quantum AI, Venice, CA 90291, USA}


\date{\today}

\begin{abstract}
The state hidden subgroup problem (StateHSP) is a recent generalization of the hidden subgroup problem.
We present an algorithm that solves the non-abelian StateHSP over $N$ copies of the dihedral group of order $8$ (the symmetries of a square).
This algorithm is of interest for learning non-Pauli stabilizers, as well as related symmetries relevant for the problem of Hamiltonian spectroscopy.
Our algorithm is polynomial in the number of samples and computational time, and requires only constant depth circuits. This result extends previous work on the abelian StateHSP and, as a special case, provides a solution for the ordinary hidden subgroup problem on this specific non-abelian group.
\end{abstract}


\maketitle


\section{Introduction}

The hidden subgroup problem (HSP) is a paradigmatic problem to attack with quantum algorithms, hosting most known exponential speed-ups \cite{NC_2019}.
Recently such a speed-up has even been realized in a proof-of-principle experiment \cite{Singkanipa_PRX_HSP_expt}.
For abelian groups the problem is well understood~\cite{childs_quantum_2010}, though some non-abelian groups have also been tackled (see e.g. Refs.~\cite{Ivanyos_arXiv_2007_nonabelian_example_I, Imran_arXiv_2023_nonabelian_example_II, alagic_arXiv_2007_general_simon, Bacon_arXiv_2008_dihedral_simon,FIMSS_2003_translating_coset, Kuperberg_arXiv_2004_sieve, kuperberg_arXiv_2011_sieve_II, goncalves_arXiv_2011_some_semidirect_HSPs, Bacon_childs_vanDam_IEEE_2005_more_groups}).
General non-abelian groups appear to be hard, however, and solutions for these groups would solve some related problems that are thought to be hard~\cite{childs_quantum_2010}.
An exciting new direction of HSP-related work has emerged recently with the StateHSP variant of the problem~\cite{Bouland_arXiv_2024_stateHSP, Hinsche_arXiv_2025_Abelian_stateHSP}.

\begin{defn}[State Hidden Subgroup Problem (StateHSP) \cite{Bouland_arXiv_2024_stateHSP}]\label{defn:stateHSP}
        Let $G$ be a finite group, and $\rho: G \rightarrow V$ be a unitary representation of $G$. Suppose we are given a subgroup $H \leq G$ and copies of a state $|\Psi \rangle \in V$ such that
        \begin{enumerate}
            \item For any $h \in H$, we have $\rho(h) | \Psi \rangle = | \Psi \rangle$, 
            \item For any $g \notin H$, we have $| \langle \Psi | \rho(g) | \Psi \rangle | < 1 - \varepsilon$, for some constant $\varepsilon$.
        \end{enumerate} 
        The problem is to identify $H$.
\end{defn}

In the HSP, we are usually presented with a classical function that hides the subgroup $H$. \cref{defn:stateHSP} generalizes this in two ways. First, the StateHSP may take as input quantum states expressing a symmetry via arbitrary representations of the group. Second, the StateHSP only requires an $\varepsilon$ allowance with respect to other symmetries. This is in contrast to the HSP, which can be formulated as a StateHSP, but only on the regular representation, and with $\varepsilon = 1$ \cite{Bouland_arXiv_2024_stateHSP}. The broadening of the scope of the HSP has permitted applications to several physically relevant problems, including the hidden cut problem \cite{Bouland_arXiv_2024_stateHSP}, and the stabilizer learning problem \cite{Hinsche_arXiv_2025_Abelian_stateHSP}. Both of these are instances of \textit{abelian} StateHSPs. 

\begin{fact}[Abelian StateHSP \cite{Hinsche_arXiv_2025_Abelian_stateHSP}]\label{fact:abelian_stateHSP}
    For any finite abelian group $G$ and $\varepsilon > 0$, there exists a polynomial-time algorithm that, with $O(\log |G| / \varepsilon)$ copies of $|\Psi \rangle$, identifies the hidden subgroup $H \leq G$. 
\end{fact}

In the abelian case, all information about the hidden subgroup is contained in the Fourier sampling distribution (up to classical processing), hence the algorithm for the abelian StateHSP is simply proceeds via Fourier sampling until a solution is found \cite{Hinsche_arXiv_2025_Abelian_stateHSP}. However, the same strategy does not work for non-abelian groups in general \cite{Bouland_arXiv_2024_stateHSP, alagic_arXiv_2005_strong_fourier_fails}; we will comment more on the specific difficulties in \cref{ssec:nonabelian_comments}. This leads to a natural question: What non-abelian StateHSPs can also be solved, and what are they good for?

\subsection*{Outline of work}

As a first step in this direction, we study the StateHSP for a particular non-abelian group: the direct product of $N$ copies of the dihedral group of order 8, $\mathcal{D}_4$ (the symmetries of a square). We refer to this as the Multiple-Squares StateHSP.
The HSP for dihedral groups provides a non-abelian problem to work on that is not too far from an abelian problem, and connected to lattice-based problems that are relevant to cryptography~\cite{regev_quantum_2003}, making the state version of the dihedral HSP and its variants a natural direction to explore.
In this work we study hidden involutions, which when phrased as a StateHSP show up naturally in the physically relevant problem of Hamiltonian spectroscopy.
In \cref{sec:motivation}, we motivate and introduce the Multiple-Squares StateHSP, and discuss connections to prior work on the ordinary HSP.
In \cref{sec:main_algo}, we present an algorithm that solves the Multiple-Squares HSP. The algorithm also heavily features Fourier sampling -- however, unlike the abelian case, performing Fourier sampling directly on copies of $| \Psi \rangle$ does not work. This is well-known already from prior work on the HSP, with no-go results arising from the fact that most irreps of $\mathcal{D}_4^N$ have exponentially large dimension \cite{alagic_arXiv_2005_strong_fourier_fails}.
Indeed, most of the work of the algorithm involves constructing abelian groups out of the problem, and the appropriate classical processing of this data.
In \cref{sec:statements}, we state the main theorem and outline the propositions enabling its proof.
In \cref{sec:proof_correctness}, we prove the correctness of the algorithm.
Finally, in \cref{sec:arbitrary_reps}, we generalize the algorithm to arbitrary representations of $\mathcal{D}_4^N$. 

\section{Setup}\label{sec:motivation}

\subsection{Problem and motivation}

We will start by formally stating our problem, before motivating several ways one might have arrived at this problem. While we focus a particular representation here, we will generalize to arbitrary representations in Sec.~\ref{sec:arbitrary_reps}.  First, recall that the Dihedral group of a square is generated by two elements, $r$ (a reflection) and $s$ (a rotation by 90-degrees),
\begin{align}
    \mathcal{D}_4=\langle r, s:r^2=s^4=e,\,rsr=s^{-1}\rangle
    \,.
\end{align}
Each group element can be uniquely represented by the product $r^ts^k$, $0\leq t\leq1$, $0\leq k\leq4$. The group of multiple squares is the direct product $\mathcal{D}_4^N := \bigotimes_{n=1}^N \mathcal{D}_4$, which has elements of the form 
\begin{equation}
g=(r^{t_1}s^{k_1},\ldots,r^{t_N}s^{k_N}) \in \mathcal{D}_4^N,
\end{equation}
and can be said to be specified by vectors $t \in \mathbb{Z}_2^N, k \in \mathbb{Z}_4^N$, with entries $t_n, k_n$. This notation is a little clunky; we pare it down later. 

For most of the text, we will be interested in the following representation, which we will refer to as the \textit{doubled representation}. This is the (unitary) representation $U_2$ of $\mathcal{D}_4$ on two qubits $A, B$, 
\begin{equation}
    U_2(r^t s^k) = X^t_A e^{i \pi k Z_A / 4} \otimes X^t_B e^{i \pi k Z_B / 4},
\end{equation}
where the subscript of the Paulis $X, Z$ indicate which qubit they act on. This is easily promoted to a representation of $\mathcal{D}_4^N$ on $2N$ qubits, labelled $(n, A/B)$, by taking tensor products, 
\begin{equation}
\begin{aligned}
&U_2^N(r^{t_1}s^{k_1},\ldots,r^{t_N}s^{k_N}) 
= \bigotimes_{n=1}^N U_2(r^{t_n} s^{k_n}) \\
&= \bigotimes_{n=1}^N  X^{t_n}_{n, A} e^{i \pi k_n Z_{n, A} / 4} \otimes X^{t_n}_{n, B} e^{i \pi k_n Z_{n, B} / 4}.
\end{aligned}
\end{equation}
We will also refer to this as the doubled representation; with the number of qubits $1$ or $N$ specified by context. We refer to each $n = 1, ..., N$ as a `site', which comprises the two qubits $(n, A/B)$. The fact that the doubled representations are in fact representations may be verified via direct computation. With this we are now able to formally state our first StateHSP.

\begin{defn}[Multiple squares StateHSP I]\label{defn:ms-shsp}
    Let $h=(r^{t^h_1}s^{k^h_1},\ldots,r^{t^h_N}s^{k^h_N})$ be an element of $\mathcal{D}_4^N$ such that $h^2 = e$ ($h$ is an involution).
    Assume that you have access to copies of an unknown quantum state vector $|\Psi\rangle\in\mathbb{C}^{2N}$ that is promised to have the following properties:
    \begin{enumerate}
        \item $U_2^N(h)|\Psi\rangle=|\Psi\rangle$.
        \item $\forall g \notin \{e, h\}: |\langle\Psi|U_2^{N}(g)|\Psi\rangle|\leq1-\varepsilon$.
    \end{enumerate}
    The problem is to identify $h$.
\end{defn}
Note that we have restricted the subgroup to be a two-element subgroup generated by an involution and we are restricting ourself to the particular representation $U_2^N$ -- we will state a more general version of this problem in \cref{sec:arbitrary_reps}.

At first glance, the Multiple-Squares group, along with the doubled representation seems to have sprung out of nowhere.
While the HSP is tied to a classical-subgroup-hiding function on the regular representation, the more general statement of the StateHSP allows one to directly address the possibility of quantum inputs, which leads to the application of HSP techniques to various physical tasks. The seminal work of Ref.~\cite{Bouland_arXiv_2024_stateHSP} identifies the hidden cut problem as one such task, and Ref.~\cite{Hinsche_arXiv_2025_Abelian_stateHSP} shows that all Pauli stabilizer learning is another such task. Neither of these tasks directly admits a description in terms of a classical subgroup hiding function on the regular representation. In the same spirit, we will now motivate and apply our problem to various physically relevant scenarios. The already-motivated reader may skip directly to \cref{sec:main_algo}, where we describe the solution to this problem.

\bigskip

\textit{(a) Stabilizer learning -- } The simplest situation in which the Multiple-Squares group arises is a variant of stabilizer learning \cite{Montanaro_arXiv_2017_stabilizer_bell, Hinsche_arXiv_2025_Abelian_stateHSP, GIKL24_improved_stab_est}. In the StateHSP version of stabilizer learning, we are given a state $| \Psi \rangle$ on $N$ qubits that we are told is stabilized by some $S \in \mathcal{P}_N$ from the Pauli group, such that $ S | \Psi \rangle = \pm | \Psi \rangle$. For any other $S' \in \mathcal{P}_N$, we are promised that $| \langle \Psi | S' | \Psi \rangle | \leq 1 - \varepsilon$. For simplicity, we consider the case of a single stabilizer $S$. The sign of the stabilizing Pauli can be fixed by doubling the state, which then formally takes the form of an abelian StateHSP.

How can one go beyond the Pauli group, and introduce non-Clifford-ness into the problem?
One simple way to do so is to take the Pauli group and conjugate it by $T$ gates in arbitrary locations, which increases the number of candidate stabilizers by a factor of $2^N$, the number of possible distributions of $T$ gates.

To be explicit, the candidate stabilizer in the Pauli stabilizer learning problem is a string of operators drawn from $\{I, X, Y, Z\}$. Now we additionally have $T X T^{\dag} = (X +  Y)/\sqrt{2}$ and $T Y T^{\dag} = (X - Y)/\sqrt{2}$.
We get nothing new from $TZT^{\dag} = Z$. Hence, in this modified version of the problem the candidate stabilizer is instead a string of operators drawn from the set $\{I, X, Y, Z, (X \pm Y)/\sqrt{2}\}$. As before, the sign of the stabilizing operator is fixed by doubling the state -- in that case the candidate doubled stabilizers are in fact exactly the doubled representations of involutions of the Multiple-Squares group, and one obtains \cref{defn:ms-shsp}.
We note that these non-Pauli stabilizers are relevant in quantum error correction, showing up for example in quantum double models hosting non-abelian topological order \cite{Kitaev_Annals_Physics_2003}, as well as more general non-Pauli stabilizer codes that can host transversal non-Clifford gates \cite{Webster_Quantum_2022_XP_formalism}.

This problem also shares some similarities with the problem of learning T-doped stabilizer states \cite{Leone_quantum_2024_T_doped_I, Chia_quantum_2024_T_doped_II}.
This is the problem of learning a state prepared by a Clifford circuit doped with some number of non-Clifford gates, which is another way to introduce non-Clifford-ness into the problem.
Because of the $\varepsilon$ requirement, most states fulfilling the StateHSP promise are quite far from these states, and in the StateHSP we are not required to learn the state, just the symmetry.
However, solving an abelian StateHSP to identify the maximal stabilizer symmetry of a T-doped stabilizer state is the first step of the state-learning algorithm, highlighting that generally these symmetry-learning problems are important components in general quantum learning algorithms~\cite{huang_quantum_2022}.

\bigskip

\textit{(b) Hamiltonian spectroscopy -- } Given a Hamiltonian, one often wants to obtain information about its spectrum~\cite{dalmonte_quantum_2018,dong_ground_2022}.
This task is greatly simplified given knowledge of the symmetries of the Hamiltonian, as this reduces the degrees of freedom of the problem and suggests the preparation of ideal probe states~\cite{tubman_postponing_2018,erakovic_high_2025}.
For instance, knowing that a Hamiltonian exhibits translational invariance suggests that one should probe it with momentum eigenstates.

Recall that a symmetry of a Hamiltonian $H$ is given by some operator that commutes with $H$, taken from a projective representation of some group on the Hilbert space on which $H$ acts. The principles in the following discussion apply to any group. However, for concreteness, let us assume $H$ acts on $N$ qubits, and the symmetry operator is of the form
\begin{equation}
    U(t^h, k^h) := \bigotimes_{n=1}^N X_n^{t_n^h} e^{i \pi k_n^h Z_n / 4},
\end{equation}
where we use the superscript $h$ to denote a particular fixed vector, and with the Hamiltonian satisfying 
\begin{equation}
    [H, U(t^h, k^h)] = 0,
\end{equation}
for a unique pair $(t^h, k^h)$. Clearly, the two-fold tensor product of $U(t^h, k^h)$ yields an element of the doubled representation of the Multiple-Squares group. Setting $h = r^{t^h_1} s^{k^h_1}, ..., r^{t^h_N} s^{k^h_N}$, we find the above also implies, 
\begin{equation}
    [H \otimes I_B, U_2^N(h) ] = 0,
\end{equation}
where we say that $H$ acts on the $A$ qubits, and $I_B$ on the $B$ qubits. 

In the above setup, we clearly have a hidden group element. How do we turn it into a StateHSP? The Bell states will play a key role in our algorithm later, but for now, we observe that the Bell state
\begin{equation}
    | \Psi^+ \rangle = \frac{1}{\sqrt{2}}(| 01 \rangle + | 10 \rangle),
\end{equation}
is simultaneously a $+1$ eigenstate of $U_2(r s^k), k = 0, 1, 2, 3$, $U_2(s^2)$, and trivially the identity. Now, an involution of $\mathcal{D}_4^N$ has to be a direct product of elements of those forms. Hence, the $N$-fold tensor product of Bell states $| \Psi^+ \rangle^{\otimes N}$, where the $n$-th Bell state is supported on the $(n, A)$ and $(n, B)$ qubits, is simultaneously a $+1$ eigenstate of all involutions in the Multiple-Squares group.

Finally, we may evolve $| \Psi^+ \rangle^{\otimes N}$ under the Hamiltonian $H$. Because of the symmetry, we have
\begin{align}
    &U_2^{N}(h)\exp(-it\,H\otimes I_B)|\Psi^+\rangle^{\otimes N} \\
    &=
    \exp(-it\,H\otimes I_B)|\Psi^+\rangle^{\otimes N}
    \,,
\end{align}
which is the first requirement for a StateHSP, with $|\Psi\rangle=\exp(-it\,H\otimes I_B)|\Psi^+\rangle^{\otimes2}$.
For rich enough Hamiltonian evolution, which we leave vague in this motivational sketch, the state $\exp(-it\,H\otimes I^{\otimes N})|\Psi^+\rangle^{\otimes N}$ will not be invariant under the action of any other $g\in\mathcal{D}_4^N$, satisfying the second requirement for a StateHSP.

The general lesson is that Hamiltonian symmetry learning can be cast as a StateHSP, motivating yet another reason to study the StateHSP. This is subject to the availability of `good' probe states; these being states such as $| \Psi^+ \rangle^{\otimes N}$, which start out as simultaneously $+1$ eigenstates of all the candidate stabilizers. Other simple examples include the state $| + \rangle^{\otimes N}$ if we are promised that the hidden symmetry is an $X$ string, or tensor products of Bell states if the hidden symmetry is Pauli type (up to a potential choice of sign).

\bigskip

\textit{(c) The simplest non-abelian StateHSP -- } Finally, \cref{defn:ms-shsp} may be motivated directly from the study of StateHSPs. In particular, since all abelian StateHSPs are solved, one is led to ask what the simplest non-abelian StateHSP is, in order to capture the differences between abelian and non-abelian StateHSPs, and ask which lessons from abelian StateHSPs carry over to the non-abelian case. In several ways, the doubled representation of the Multiple-Squares group presents the minimal elements that make up an interesting non-abelian StateHSP. First, the group $\mathcal{D}_4$ is a minimal example of a non-abelian group in the sense that it only has a single non-1D irrep. Next, taking direct products allows us to obtain a group with some sense of an extensive scaling, while retaining the simple structure of $\mathcal{D}_4$. Finally, considering the doubled representation allows us to make sense of a representation other than the regular representation. Importantly, $\mathcal{D}_4^N$ has the feature that most irreps are exponentially large, which prevents one from directly using Fourier sampling to solve the problem \cite{alagic_arXiv_2005_strong_fourier_fails}.

\bigskip

Next, as background, we will provide some comments on the the differences between the HSP and StateHSP, as well as connections to prior art on the Multiple-Squares HSP.

\subsection{HSP vs. StateHSP: The non-abelian case}\label{ssec:nonabelian_comments}

Before embarking on our problems, we provide some remarks on the fact that information-theoretically, the StateHSP is solvable in a polynomial (in $|G|$) number of samples, by recourse to an information-theoretic solution for the HSP. If we assume access to a controlled-group action as well as an ancillary register encoding the regular representation of the relevant group, and  one can reproduce the traditional coset states of the HSP to exponential (in say $t$) accuracy using $O(t)$ copies of the state supplied by the StateHSP, via an orthogonality amplification technique \cite{FIMSS_2003_translating_coset, Bouland_arXiv_2024_stateHSP}. Then, one can apply the techniques of Ref.~\cite{Ettinger_2004_HSP_polynomial} to obtain the hidden subgroup. The algorithm of Ref.~\cite{Ettinger_2004_HSP_polynomial} only requires a polynomial number of samples (but exponential computational time), so with the appropriate resources, this yields an efficient information-theoretic solution to any StateHSP. 

One can go a step further for Abelian StateHSPs. In this case, the resulting probability distribution from weak Fourier sampling suffices to solve the HSP in polynomial samples and time. The coset states obtained from copies of the StateHSP state can be used to sample from this distribution. Hence,  up to a $|G|$ independent overhead, one may obtain a computationally efficient solution for any abelian StateHSP. We may think about this as an indirect solution, given by simulating the HSP with the StateHSP. Conversely, this work is concerned with a direct solution that does not make calls to the HSP.

Since all StateHSPs can be solved by recourse to the HSP, why bother with direct solutions? First, on a fundamental level, one is interested in asking what the minimal resources required to solve a StateHSP -- eg. when can the StateHSP be solved without access to the regular representation, or the expensive multi-controlled operations required to perform the controlled-group action? On a practical level, removing some of this complexity moves the implementation of such algorithms closer to the present. For instance, for the particular StateHSP studied in this work, we show that one only ever needs two-qubit operations/measurements and constant depth circuits, which is far simpler than if we were to try to solve it via the indirect method.

As a final observation motivating the study of non-abelian StateHSPs, we note that many non-abelian HSPs are not expected to admit efficient solutions. However, this does not rule out the possibility that there are particular non-trivial representations of such groups, with relevant physical applications, for which the StateHSP yields efficient solutions. An exciting line of research is to study the relationship between representations of groups and efficient solution of the StateHSP for non-abelian HSPs that do not admit efficient solutions. We leave the investigation of these groups to future work. In this work, we will show regardless of representation, $\mathcal{D}_4^N$ admits a computationally efficient solution. However, we do not expect such representation-independent statements to be possible for most non-abelian groups.

\subsection{Connection to previous solutions for the regular HSP}

Finally, we remark on how our algorithm compares to previous solutions which may be applied to the HSP on the Multiple-Squares group \cite{FIMSS_2003_translating_coset, alagic_arXiv_2007_general_simon, Bacon_arXiv_2008_dihedral_simon}. The reader may find this discussion easier to digest after skimming the details of the main algorithm in \cref{sec:main_algo}. We begin by emphasizing that despite strong analogies, \textit{previous solutions are not a priori directly applicable to our problem}. This is due to the key differences between the HSP and the StateHSP, along with the fact that weak Fourier sampling alone does not suffice to efficiently solve a non-abelian StateHSP. Conversely, one may try to identify techniques from HSP solutions that are independent of features such as the uniformity of sampling coset states or the nature of the regular representation.

As a first example, Ref.~\cite{FIMSS_2003_translating_coset} solves the HSP for a large class of non-abelian groups, which includes $\mathcal{D}_4^N$, in quasi-polynomial time.
A key step in their algorithm involves building states which are uniform superpositions of particular group elements, such that the identification of the stabilizers of these states helps us make progress towards identifying the hidden subgroup.
In the StateHSP, this is defied by the fact that most representations do not admit such superpositions, since we cannot identify group elements with state vectors.
On the other hand, the algorithm of Ref.~\cite{FIMSS_2003_translating_coset} heavily utilizes the quotient groups that turn up in the derived series of their groups.
Along similar lines, we use the fact that $\mathcal{D}_4^N$ has a very short derived series (of length $2$), which gives us an abelian quotient group $\mathcal{D}_4^N/Z[\mathcal{D}_4^N]$, and in fact Step 2. of our algorithm involves constructing and then Fourier sampling over this group.

Finally, studying the Multiple-Squares HSP, the algorithm of Ref.~\cite{Bacon_arXiv_2008_dihedral_simon} provides a solution that works with $O(N^3)$ copies of the state. When applied to the HSP (see \cref{sec:arbitrary_reps}), our algorithm works with $O(N^2)$ copies of the state, providing a mild improvement in sample complexity.
This difference comes from the fact that Ref.~\cite{Bacon_arXiv_2008_dihedral_simon} (as well as Ref.~\cite{alagic_arXiv_2007_general_simon}), require the algorithm to first identify the locations of all non-trivial reflections.
In contrast, we show that this step can be postponed to the end, at which point it becomes much simpler. We also remark that while the sample savings seems small, perhaps a more important advantage of our algorithm is that we remove almost all the need for coherent control (of more than two qubits at once), leading to only constant depth circuits, reducing greatly the resource requirements to run this on an actual quantum device.

\section{Representation preliminaries}\label{sec:dihedral_prelims}

In this section, we collect some facts about the group $\mathcal{D}_4$ that will be useful for the technical exposition later.  

\bigskip

\textit{Reps and irreps of the $\mathcal{D}_4$ -- } We will often reference the two-dimensional representations
\begin{align}
\label{eq:dihedral-2d-reps}
    D^j(r^ts^k)
    &=
    X^t(iZ)^{kj}
    &
    j
    &\in
    \{0,1,2,3\}
\end{align}
and the one-dimensional representations
\begin{align}
    D^T(r^ts^k)
    &=
    1
    &
    D^A(r^ts^k)
    &=
    (-1)^t
    \\
    D^R(r^ts^k)
    &=
    (-1)^k
    &
    D^{RA}(r^ts^k)
    &=
    (-1)^{k+t}
    \,,
\end{align}
The two dimensional representations $D^1$ and $D^3$ are irreducible (and unitarily equivalent), while the others decompose as $D^0=D^T\oplus D^A$ and $D^2=D^R\oplus D^{RA}$.
The tensor product of the two dimensional irrep with itself also decomposes into one-dimensional irreps: $D^1\otimes D^1=D^T\oplus D^A\oplus D^R\oplus D^{RA}$.

We will primarily use the notation $k=2v+w$ for the power of the primitive rotation $s$.
This notation has the advantage that both $v$ and $w$ are bits, and that $w$ labels the conjugacy class of the reflection $rs^{2v+w}$.
We will also often notate the group elements directly by these triples, $(t,v,w)$, so, for example, $D^j(r^ts^{2v+w})=D^j(t,v,w)$.

\bigskip

\textit{Irreps of the quotient group -- } One important fact is that the collection of all one-dimensional irreps is not faithful.
The joint kernel is the center $\mathbb{Z}_2$ generated by $s^2=(0,1,0)$, and the quotient is $\mathcal{D}_4/\langle s^2\rangle\cong\mathbb{Z}_2\times\mathbb{Z}_2$.
It's convenient to take elements of the form $(t,0,w)$ as the coset representatives for the quotient, so we will often notate elements of this quotient group just by the pair $(t,w)$.

The irreps of $\mathbb{Z}_2\times\mathbb{Z}_2$ are products of the irreps of $\mathbb{Z}_2$, which are
\begin{equation}
\label{eq:z2-irreps}
    \begin{aligned}
        C^0(x) &= 1,
        &
        C^1(x) &= (-1)^x.
    \end{aligned}
\end{equation}
We label the irreps of $\mathbb{Z}_2\times\mathbb{Z}_2$ as $C^{q,p}=C^qC^p$.
These are in one-to-one correspondence with the one-dimensional irreps of $\mathcal{D}_4$, 
\begin{equation}
\label{eq:z2z2-irreps}
    \begin{aligned}
        D^T(t, v, w) &= C^0(t) C^0(w) = C^{0,0}(t,w), \\
        D^A(t, v, w) &= C^1(t) C^0(w) = C^{1,0}(t,w), \\
        D^R(t, v, w) &= C^0(t) C^1(w) = C^{0,1}(t,w), \\
        D^{RA}(t, v, w) &= C^1(t) C^1(w) = C^{1,1}(t,w). \\
    \end{aligned}
\end{equation}
We can notate this succinctly as $C^{q,p}(t,w)=(-1)^{q\cdot t+p\cdot w}$, which also calls to mind the symplectic product between $(p,q)$ and $(t,w)$ from the stabilizer formalism.

Because $D^0$, $D^2$ and $D^j\otimes D^k$ for $j,k\in\{1,3\}$ decompose into one-dimensional irreps, and therefore cannot be faithful, they also define representations of the quotient $\mathbb{Z}_2\times\mathbb{Z}_2$:
\begin{align}
    \label{eq:conjugacy-representations}
    D_\mathcal{C}^{\{j,k\}}(t,w)
    &=
    (D^j\otimes D^k)(t,v,w)
    &
    j,k
    &\in
    \{1,3\}
    \\
    D_\mathcal{C}^{\{\ell\}}(t,w)
    &=
    D^\ell(t,v,w)
    &
    \ell
    &\in
    \{0,2\}
    \,,
\end{align}
which are well defined since the RHS does not depend on what value of $v$ is supplied. Hence the subscript $\mathcal{C}$ alludes to the fact that we will later use this representation to pick out the conjugacy class of the reflection associated with the hidden involution. 
We can explicitly diagonalize these representations using the eigenstates $|\pm\rangle$ of the $X$ operator and the Bell states
\begin{align}
    \label{eq:bell-state-defn}
    |\Phi^\pm\rangle
    &=
    (|00\rangle\pm|11\rangle)/\sqrt{2}
    \\
    |\Psi^\pm\rangle
    &=
    (|01\rangle\pm|10\rangle)/\sqrt{2}
    \,.
\end{align}
The result, making use of the $\mathbb{Z}_2\times\mathbb{Z}_2$ characters $C^{q,p}$, is
\begin{align}
    \label{eq:irrep-bases}
    D_{\mathcal{C}}^{\{1,1\}}
    &=
    D_{\mathcal{C}}^{\{3,3\}} \\
    &= \nonumber
    |\Psi^+\rangle\langle\Psi^+|C^{0,0}
    +|\Phi^+\rangle\langle\Phi^+|C^{0,1} \\
    &\qquad
    +|\Psi^-\rangle\langle\Psi^-|C^{1,0}
    +|\Phi^-\rangle\langle\Phi^-|C^{1,1}
    \label{eq:d1-d1-irrep-decomp}
    \\
    D_{\mathcal{C}}^{\{1,3\}}
    &=
    D_{\mathcal{C}}^{\{3,1\}} \\
    &= \nonumber
    |\Psi^+\rangle\langle\Psi^+|C^{0,1}
    +|\Phi^+\rangle\langle\Phi^+|C^{0,0} \\
    &\qquad +|\Psi^-\rangle\langle\Psi^-|C^{1,1}
    +|\Phi^-\rangle\langle\Phi^-|C^{1,0}
    \\
    D_{\mathcal{C}}^{\{0\}}
    &=
    |+\rangle\langle+|C^{0,0}+|-\rangle\langle-|C^{1,0}
    \label{eq:d0-irrep-decomp}
    \\
    D_{\mathcal{C}}^{\{2\}}
    &=
    |+\rangle\langle+|C^{0,1}+|-\rangle\langle-|C^{1,1}
    \,.
\end{align}

\bigskip

\textit{Parity measurements -- } We will make significant use of parity measurements, which we define by the partial isometries
\begin{equation}
\begin{aligned}
\label{eq:parity-projections}
    \Pi_0
    &=
    |0\rangle\langle01|+|1\rangle\langle10| 
    \\
    \Pi_1
    &=
    |0\rangle\langle00|+|1\rangle\langle11|
    \,
\end{aligned}
\end{equation}
corresponding exactly to projection onto eigenspaces of the parity $Z^{\otimes 2}$ operator, followed by a relabelling of basis elements. Note that in this notation, the even parity outcome is associated (in the sense of a subscript of $\Pi$) with $1$ and the odd parity outcome is associated with $0$.

The notation has been chosen such that for the doubled representation on a single pair of qubits, we have,
\begin{align}
    \Pi_jU_2(g)\Pi_j^\dagger
    &=
    D^j(g),
\end{align}
which can be verified by direct computation.

We now consider the doubled representation across $N$ pairs of qubits. Let $\pi \in \mathbb{Z}_2^N$; we can define a parity projector across $N$ pairs via,
\begin{align}
    \Pi_\pi
    &=
    \bigotimes_j\Pi_{\pi_j}
    \,.
\end{align}
and we say that $\Pi_{\pi}$ projects us onto a parity subspace $\pi$.

Projecting the doubled representation onto a parity subspace naturally yields the representation,
\begin{align}
\label{eq:parity-rep}
    D^\pi
    &=
    \bigotimes_{n=1}^ND^{\pi_n}.
\end{align}
We will often consider many copies of the doubled representation, and their projections onto the parity subspace. In this case, we denote the representation on the $(M+1)N$ qubits hosting the set of $M+1$ parity outcomes $\{\pi^m\}_{m=0}^M$ as
\begin{align}
\label{eq:parity-ensemble-rep}
    D^{\{\pi^m\}_{m=0}^M}
    &=
    \bigotimes_{m=0}^MD^{\pi^m}
    \,.
\end{align}
Where clear, we will omit the $m=0, M$ super/subscripts, i.e. we write $D^{\{\pi^m\}}$. We will also notate the normalized, parity-projected state as
\begin{align}
    |\Psi(\pi)\rangle
    &=
    \frac{\Pi_\pi|\Psi\rangle}{\Vert\Pi_\pi|\Psi\rangle\Vert}
\end{align}
and the collection of such states for a set of parity outcomes $\{\pi^m\}_{m=0}^M$ as
\begin{align}
    |\Psi(\{\pi^m\})\rangle
    &=
    \bigotimes_{m=0}^M|\Psi(\pi^m)\rangle
    \,.
\end{align}

\textit{Phaseless Pauli subgroup -- } Finally, we note that the phaseless Pauli group, which is isomorphic to $\mathbb{Z}_2 \times \mathbb{Z}_2$, is isomorphic to an abelian subgroup of $\mathcal{D}_4$. These are the elements of the form $(t, v, 0) \in \mathcal{D}_4$. This obviously corresponds to doubled Paulis in the doubled representation, but an analogous statement holds also for any representation of $\mathcal{D}_4^N$. We emphasize that this is not the same as quotient $ \mathcal{D}_4 / \langle s^2 \rangle \cong \mathbb{Z}_2 \times \mathbb{Z}_2$.

For this subgroup, we may perform a decomposition of the reps in \cref{eq:dihedral-2d-reps} into 1D irreps of $\mathbb{Z}_2 \times \mathbb{Z}_2$ as follows,
\begin{equation}\label{eq:pauli_irreps}
    \begin{aligned}
        D^0(t, v, 0) &= | + \rangle \langle + | C^0(t) C^0(v) + | - \rangle \langle - | C^1(t) C^0(v), \\
        D^1(t, v, 0) &= | + \rangle \langle + | C^0(t) C^1(v) + | - \rangle \langle - | C^1(t) C^1(v), \\
        D^2(t, v, 0) &= | + \rangle \langle + | C^1(t) C^0(v) + | - \rangle \langle - | C^0(t) C^0(v), \\
        D^3(t, v, 0) &= | + \rangle \langle + | C^1(t) C^1(v) + | - \rangle \langle - | C^0(t) C^1(v),
    \end{aligned}
\end{equation}
and we see that Fourier sampling in this group can be carried out by first sampling the reps of \cref{eq:dihedral-2d-reps} and then performing an $X$ measurement.

\begin{widetext}

\section{Algorithm}
\label{sec:main_algo}

We now present our algorithm for solving the Multiple-Squares StateHSP. In \cref{ssec:algo_spec}, we specify the steps of the algorithm.
To streamline the presentation we leave a number of subroutines undefined in the statement of the algorithm in \cref{ssec:algo_spec}; these are typeset in bold text and their definitions follow in \cref{ssec:subroutine}.

\subsection{The Multiple-Squares StateHSP algorithm}\label{ssec:algo_spec}

\begin{enumerate}
    \item Perform \textbf{Learning Pauli Stabilizers}, in case $U_2^N(h)$ is a doubled Pauli.
    If it is a Pauli, one is finished.
    Otherwise, continue \footnote{We note that this first step is only required so that we may assume that $t^h \neq 0$ everywhere in subsequent steps. In principle, this may be absorbed into the parity sampling subroutine, as this in fact performs Fourier sampling over the center of $\mathcal{D}_4^N$, and hence by \cref{fact:abelian_stateHSP} allows us to identify a hidden element of the form $(t^h, 0, 0)$ if it exists. In practice, this entails performing parity sampling $O(N/\varepsilon)$ times and finding that a full basis is not obtained. In the main text, we include an explicit Learning Pauli Stabilizers step for conceptual clarity.}.
    \item Determine $w^\text{max}$, which gives the conjugacy class of the reflections on each factor in the hidden element~$h$.
    This is done in two steps:
    \begin{enumerate}
        \item Perform Fourier sampling on the $\mathbb{Z}_2^N\times\mathbb{Z}_2^N$ quotient of $\mathcal{D}_4^N$ by its center to generate the annihilator subgroup $K^\perp$.
        Samples $(q,p)$ are obtained by
        \begin{enumerate}
            \item Repeating the \textbf{parity sampling subroutine} until one accumulates a \textbf{Bell-resolvable set} that includes the first sample (requiring independent sampling of the first parity outcome ensures sufficiently rich Bell-resolvable sets).
            \item Performing the \textbf{Bell-resolution subroutine}.
        \end{enumerate}
        \item Obtain $w^\text{max}$ from $K^\perp$ using the \textbf{maximal rotation subroutine}.
    \end{enumerate}
    \item Use $w^\text{max}$ to determine where to apply $T$ gates to $|\Psi\rangle$, namely apply $T_{n, A} \otimes T_{n, B}$ in locations $n$ where $w^\text{max}_n = 1$.
    This converts $|\Psi\rangle$ to something stabilized by a Pauli, namely $U_2^N(t^h, v^h, 0)$.
    \item Learn the stabilizing Pauli, given by $U_2^N(t^h, v^h, 0)$, with \textbf{Learning Pauli Stabilizers}. This lets us learn $t^h, v^h$.
    \item Compute $w^h = t^h \odot w^{\rm max}$ to obtain $w^h$. 
\end{enumerate}

\end{widetext}

\subsection{Subroutines}\label{ssec:subroutine}

We now define the various subroutines that are used in the algorithm. We focus on their inputs and outputs, with the goal in mind that a reader having read this section and the last may be able to perform this algorithm in the lab. We remark that where unavoidable, we have made use of notation established in \cref{sec:dihedral_prelims}.

The first subroutine, Learning Pauli Stabilizers, is provided in \cite{Hinsche_arXiv_2025_Abelian_stateHSP, Montanaro_arXiv_2017_stabilizer_bell}. We simply quote the problem in StateHSP form, and present the existence of such a subroutine as corollary to \cref{fact:abelian_stateHSP}:
\begin{cor}[Learning Pauli Stabilizers]
\label{cor:LPS}
    Let $| \Psi \rangle$ be a state that is stabilized by a doubled Pauli $P \otimes P$, i.e. $P {\otimes} P \, | \Psi \rangle = \pm | \Psi \rangle$, and for any $P' \neq P$, satisfies $| \langle \Psi | P' \otimes P' | \Psi \rangle | \leq 1 - \varepsilon$. Then, there exists an algorithm to determine $P$ with high probability using $O(N/\varepsilon)$ samples, where the only quantum resource used is Bell sampling.
\end{cor}

The next subroutine is motivated by the fact that $Z_{n, A} \otimes Z_{n, B}$, where $A, B$ label the qubits associated with the $n$-th site, for any $n$ commutes with the representation $U^N_2$. 

\begin{defn}[Parity Sampling Subroutine]\label{defn:parity_sampling}
  The parity sampling subroutine takes as input a state $ | \Psi \rangle $ on $2N$ qubits, where the qubits are labelled $(n, A), (n, B)$ for $n = 1,..., N$, and outputs a bit string $\pi \in \mathbb{Z}_2^N$ and state $|\Psi(\pi) \rangle$ with probability $P(\pi)$, via the following steps:
    \begin{enumerate}
        \item Perform transversal measurements of $Z_{n,A} \otimes Z_{n, B}, n = 1,...,N$. 
        \item Assign $\pi_n = 1$ wherever we record a $+1$ outcome, and $0$ whenever we record a $-1$ outcome.
        \item The state $| \Psi (\pi) \rangle$ is the post measurement state on $N$ qubits obtained by discarding the measurement register, as per \cref{eq:parity-projections}.
    \end{enumerate}
   The form of $| \Psi (\pi) \rangle$ and the distribution $P(\pi)$ are specified by the usual laws of quantum mechanics.
\end{defn}

The parity sampling subroutine is schematically depicted in \cref{fig:Bell_resolution_fig}a. The outputs and the representations that emerge from parity sampling have been established in \cref{sec:dihedral_prelims}. The particular set of outputs we care about are outputs that form a Bell-resolvable set, defined as follows:
\begin{defn}[Bell-resolvable set]\label{defn:Bell_resolvable_set}
  A Bell-resolvable set is a set of bitstrings $\{\pi^m \in \mathbb{Z}_2^N\}_{m=0}^M$, obtained from the parity sampling subroutine, which satisfy $\sum_{m=1}^M \pi^m = \pi^0$, together with the post-measurement state.
\end{defn}

Note that following the discussion of parity measurements in \cref{sec:dihedral_prelims}, the post-measurement state is $\bigotimes_{m=0}^M | \Psi(\pi^m) \rangle$, and the projection implied by a Bell-resolvable set yields the representation $D_\mathcal{C}^{\{\pi^m\}_{m=0}^M}$ defined in \cref{eq:conjugacy-representations}.

Given a Bell-resolvable set, one may perform the Bell-resolution subroutine.
The Bell-resolution subroutine makes use of measurement in the Bell basis, defined in \cref{eq:bell-state-defn}.
\begin{defn}[Bell-resolution subroutine]\label{defn:Bell_resolution_subroutine}
  Take as input a Bell-resolvable set as defined in \cref{defn:Bell_resolvable_set}.
  For each of the $N$ sites:
  \begin{enumerate}
      \item Perform $X$ measurements on the copies with odd parity ($\pi_n=0$) outcomes, and perform Bell measurements on pairs of copies with even parity ($\pi_n=1$) outcomes.
      \item Record the measurement outcomes as a collection of bit pairs $(p_n^j,q_n^j)$ from the irrep characters $C^{q_n^j,p_n^j}$ associated with the measured states in the irrep decompositions \cref{eq:d1-d1-irrep-decomp,eq:d0-irrep-decomp}.
      \item Add these together modulo 2 to obtain the bit pair $(p_n,q_n)=(\sum_jp_n^j,\sum_jq_n^j)$
  \end{enumerate}
  Finally, concatenate together $(p_n,q_n)$ from all $N$ sites and return the pair of bit strings $(p,q)$.
\end{defn}

The Bell-resolution subroutine performs Fourier sampling on the representation of $\mathbb{Z}_2^N\times\mathbb{Z}_2^N$ given by a Bell-resolvable set, with the pair of bit strings $(p,q) \in \mathbb{Z}_2^N \times \mathbb{Z}_2^N$ labeling an irrep given by the product of the irreps $C^{q_n,p_n}$ defined in \cref{eq:z2z2-irreps}.
Since each of the $N$ sites carries a tensor product of many copies of the representations $D_\mathcal{C}^{\{1,1\}}$ and $D_\mathcal{C}^{\{0\}}$, the addition modulo 2 of $p_n^j$ and $q_n^j$ calculates the resulting combined irrep from the Clebsch Gordon rules, since $\prod_jC^{q_n^j,p_n^j}=C^{\sum_jq_n^j,\sum_jp_n^j}$.

An example of a Bell-resolvable set and the Bell-resolution subroutine is schematically depicted in \cref{fig:Bell_resolution_fig}c.

\begin{figure*}[hptb]
\centering \includegraphics[width=\textwidth]{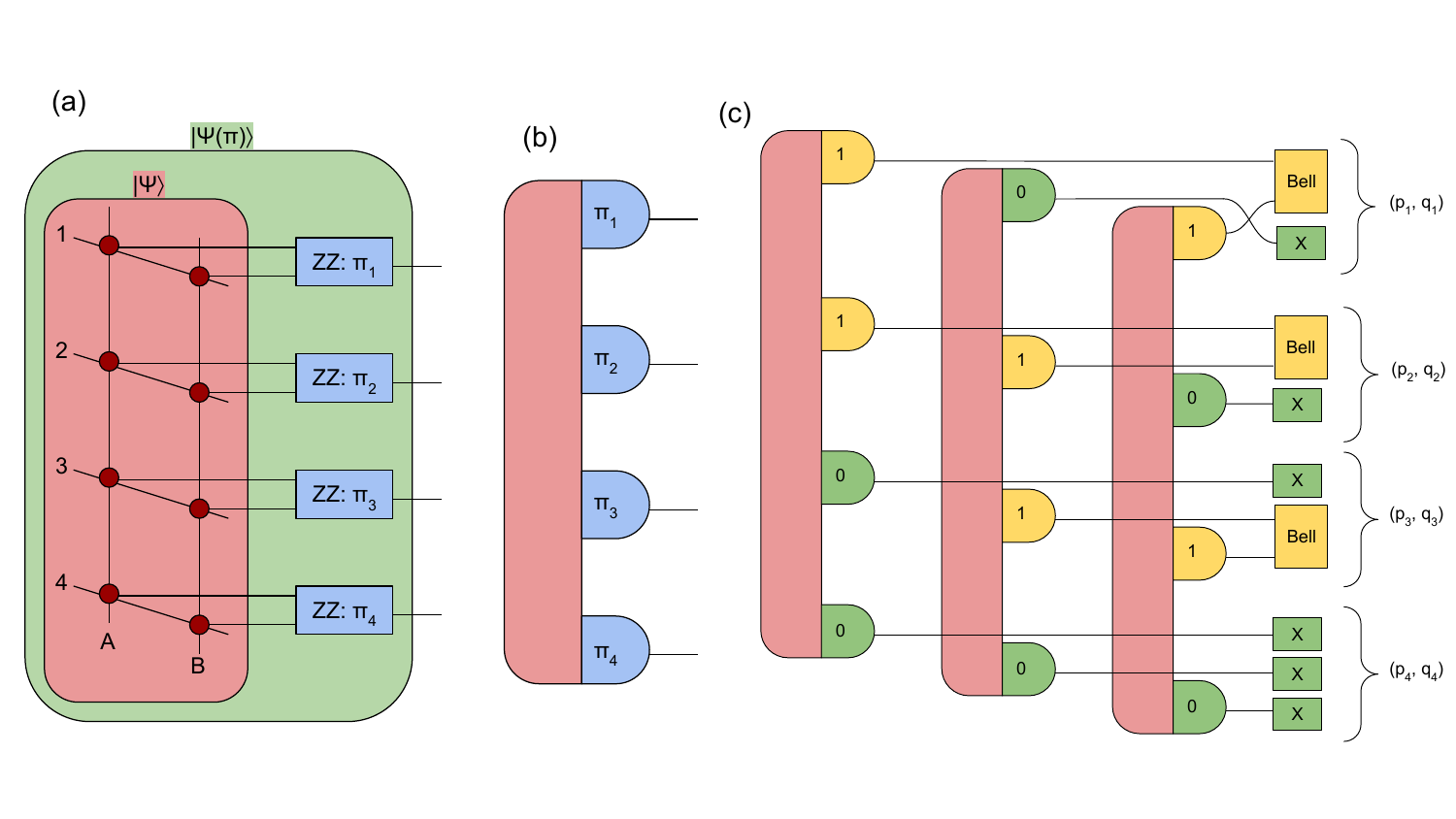} 
\caption{(a) Schematic of the parity sampling subroutine. A tranversal $Z \otimes Z$ measurement is performed on pairs of qubits, yielding an outcome $\pi$, and a post-measurement state $| \Psi(\pi) \rangle$. (b) Schematic of a projected state $| \Psi(\pi) \rangle$, where the labels on each site are given by $\pi_1, ... \pi_4$. (c) Schematic of a Bell-resolvable set, comprising three projected states with $\pi^0 = (1, 1, 0, 0), \pi^1 = (0, 1, 1, 0), \pi^2 = (1, 0, 1, 0)$. Bell measurements are performed on pairs of sites with $\pi_j = 1$, and $X$ measurements on the remaining sites. These measurement outcomes are processed to give the irrep labels $q_j, p_j$.}
\label{fig:Bell_resolution_fig} 
\end{figure*}

Taking the span of our irrep labels, treated as vectors in $\mathbb{Z}_2^N \times \mathbb{Z}_2^N$, produces the annihilator subspace $K^\perp$ of the subspace $K$ on which we perform the maximal rotation subroutine. The subspace $K$ may be obtained as the set of vectors whose inner product with members of $K^{\perp}$ is $0$, which can be obtained by arranging a basis of $K^{\perp}$ in a matrix and obtaining the nullspace via Gaussian elimination. The nullspace is $K$. We expand more on the concept of annihilator subspaces in \cref{ssec:good_nullspace_possible}.

\begin{defn}[Maximal rotation subroutine]\label{defn:maximal_rotation_subroutine}
    The maximal rotation subroutine takes as input a subspace $K\subseteq\mathbb{Z}_2^N \times \mathbb{Z}_2^N$, and outputs a vector $w^{\rm max} \in \mathbb{Z}_2^N$ via the following procedure:
    \begin{enumerate}
        \item Obtain a basis $(t^m, w^m)$ for $K$.
        \item Obtain $w^{\rm max}$ by performing bit-wise OR of the $w$ vectors across the basis.
    \end{enumerate}
\end{defn}
Note that the bitwise OR over the whole subspace gives the same result as taking the bitwise OR just over a basis, since $w^\text{max}=\bigvee_{w\in K}w=\bigvee_{w\in\text{basis}}w$.

\section{Main theorem and supporting results}\label{sec:statements}

The main theorem we will prove is:
\begin{thm}[Multiple-Squares StateHSP]\label{thm:ms-shsp}
  The Multiple-Squares StateHSP algorithm solves the Multiple-Squares StateHSP problem (\cref{defn:ms-shsp}) with probability higher than $1 - \delta$ using $O(\varepsilon^{-2} (N^2 + \log (N/\delta)))$ samples.
\end{thm}

\bigskip

\textit{Bell-resolvable sets are efficiently obtainable -- } For our overall algorithm to be efficient, it must be practical to obtain Bell-resolvable sets.
One can sample a Bell-resolvable set from a satisfactory distribution by first obtaining an initial parity outcome $\pi^0$ and then sampling a basis $\{\pi^m\}_{m=1}^N$ of additional parity outcomes.
We order the basis such that the sum of the first $M$ elements is equal to $\pi^0$, giving us the Bell-resolvable set $\{\pi^m\}_{m=0}^M$.
Sampling the basis is the only nontrivial step, which the StateHSP promise guarantees can be completed without too many samples due to anti concentration of the distribution over parities.

\begin{prop}[Sample complexity for obtaining a basis for parity]\label{prop:basis_sampling_complexity}
    Let $|\Psi \rangle$ be a state satisfying the promises of \cref{defn:ms-shsp} and let $t^h \neq 0$.
    To obtain a spanning set of parities for $\mathbb{Z}_2^N$ from the parity sampling subroutine with probability at least $1-\delta$, it suffices to collect at least
    \begin{align}
        \frac{N + \log(1/\delta)}{\varepsilon}
    \end{align}
    samples.
\end{prop}

We make several remarks about the above proposition: First, if $t^h = 0$, then the algorithm will have terminated in Step 1, as the hidden element will be a Pauli stabilizer. This justifies assuming $t^h \neq 0$. Second, we set the threshold probability to go like $\delta \varepsilon / N$, as we desire the final probability of correctness to go like $1- \delta$, and as we shall see, this subroutine must be performed $O(N/\varepsilon)$ times. This only contributes additional sub-leading factors logarithmic in $N$. Alternatively, the expected number of samples required to give a spanning set of parities can be made close to $N/\varepsilon$ with an overhead that is logarithmic in $N$, and hence will only contribute logarithmic factors in the final result.

\bigskip

\textit{Bell-resolvable sets abelianize the doubled representation -- } Why do we want to obtain Bell-resolvable sets? The answer is that on Bell-resolvable sets, our problem looks abelian; one might think about this as being analogous to how doubling a state abelianizes the Pauli group by removing phases; this allows us to make progress via Fourier sampling. The following proposition makes this precise. 

\begin{prop}[Abelianization of Multiple-Squares]\label{prop:conjugacy_ensemble_representation}
    Let $\{\pi^m\}_{m =0}^M$ be a Bell-resolvable set.
    The representation $D^{\{\pi^m\}}$ is non faithful, with kernel $\{(0,v,0):v\in\mathbb{Z}_2^N\}\cong\mathbb{Z}_2^N$.
    This representation gives a representation of the abelian quotient group $\mathbb{Z}_2^N\times\mathbb{Z}_2^N$ which we notate $D_{\mathcal{C}}^{\{\pi^m\}}$.
    This representation is defined by mapping $(t,w)\in\mathbb{Z}_2^N\times\mathbb{Z}_2^N$ to the coset representative $(t,0,w)$ and evaluating that on $D^{\{\pi^m\}}$:
    \begin{align}
        D_{\mathcal{C}}^{\{\pi^m\}}(t,w)
        &=
        D^{\{\pi^m\}}(t,0,w)
        \,.
    \end{align}
    Furthermore, for the hidden element $(t^h, v^h, w^h)$ specified in the Multiple-Squares StateHSP, we have,
    \begin{equation}
        D_{\mathcal{C}}^{\{\pi^m\}}(t^h, w^h) | \Psi(\{\pi^m\}) \rangle = | \Psi(\{\pi^m\}) \rangle
        \,.
    \end{equation}
\end{prop}

We want to emphasize that the representation $D^{\{\pi^m\}}_{\mathcal{C}}$ of $\mathbb{Z}_2^N\times\mathbb{Z}_2^N$ is \emph{not} the one obtained by restricting $D^{\{\pi^m\}}$ to the subgroup $\{(t,v,0):t,v\in\mathbb{Z}_2^N\}\cong\mathbb{Z}_2^N\times\mathbb{Z}_2^N$, but rather quotienting out the kernel of $D^{\{\pi^m\}}$.
Additionally, $D_{\mathcal{C}}^{\{\pi^m\}}(t^h, w^h)$ is by no means the unique stabilizer of $|\Psi(\{\pi^m\})\rangle$ in $\mathbb{Z}_2^N\times\mathbb{Z}_2^N$, in contrast to the original StateHSP setup. As a concrete example, one can check that if $| \Psi \rangle = |\psi \rangle \otimes |\psi \rangle$ is in tensor product form, then $| \Psi(\pi) \rangle$ is always stabilized by an $X$-type operator supported on sites where $\pi = 0$. In general, the $\varepsilon$ part of the StateHSP promise does not project down onto the states $|\Psi(\pi)\rangle$, as the contribution of each of these states to the $|\Psi \rangle$ is exponentially small.  

The fact that, for any Bell-resolvable set, $D_{\mathcal{C}}^{\{\pi^m\}}$ is a representation of the abelian group $\mathbb{Z}_2^N\times\mathbb{Z}_2^N$ allows us to leverage Fourier sampling. This is what the next step of our algorithm does. 

\begin{prop}[The Bell-resolution subroutine performs Fourier sampling]\label{prop:Bell-resolution_is_Fourier}
    The output of the Bell-resolution subroutine on a Bell-resolvable set is an irrep label $(q, p) \in \mathbb{Z}_2^N \times \mathbb{Z}_2^N$, regarded as a character. Furthermore, the hidden element $(t^h, v^h, w^h)$ specified in the Multiple-Squares StateHSP satisfies $q \cdot t^h + p \cdot w^h = 0$, again with no information about $v^h$.
\end{prop}

As mentioned, while the conjugacy ensemble representation resembles an abelian StateHSP, in any given realization, $D_{\mathcal{C}}^{\{\pi^m\}}(t^h, w^h)$ may not be a unique stabilizer. We proceed by building the common null space across all the irrep labels obtained in the Bell-resolution subroutine, but it is not \emph{a priori} clear that this would even converge to a unique answer, nor that it will happen in a polynomial number of samples.

\bigskip

\textit{Good null spaces exist -- } To get around these issues, we make the observation that at this stage, we do not actually need to narrow candidates for $(t^h, w^h)$ down to a unique null space. Instead there is a much larger class of \textit{good null spaces} that will enable the final part of our algorithm. We first define a good null space, then we explain why it is good.

\begin{defn}[Good nullspace]\label{defn:good_nullspace}
    Let $K \subseteq \mathbb{Z}_2^N \times \mathbb{Z}_2^N$ be a subgroup containing $(t^h, w^h)$. Then we say that $K$ is a \textit{good nullspace} if for all elements $(t, w)$ we have,
    \begin{equation}
      (t^h \odot w)_j = 1 \implies w^h_j = 1,
    \end{equation}
    where $\odot$ is the hadamard product. 
\end{defn}

The condition above tells us that, while $K$ can contain a lot of vectors other than $(t^h,w^h)$, all that subsequent steps in the algorithm require is that on factors $n$ where (the unknown) $t^h_n = 1$, the nullspace $K$ should not inappropriately assign $w_n=1$.
In Step 2b. of the algorithm, the nullspace obtained from the previous step is fed into the maximal rotation subroutine.
The output is characterized as follows.

\begin{lem}[Output of a good nullspace]\label{lem:good_nullspace_output}
    Let $K$ be a good nullspace. Given such a group as input into the Maximal Rotation subroutine, the output will be of the form $w^{\rm max} = w^h + (t^h)^c \odot w' \in \mathbb{Z}_2^N$, where here $(t^h)^c$ denotes the complement of $t^h$, and the form of $w'$ is unimportant.
\end{lem}

Finally, we can say why a good null space is good. The reason for this is that $w^{\rm max}$ contains enough information to allow us to deterministically transform $U^N_2(t^h, v^h, w^h) \rightarrow U^N_2(t^h, v^h, 0)$ via conjugation by $T$ gates, at which point the hidden element becomes a Pauli operator. The next proposition characterizes the effect of such a transformation.

\begin{prop}[Maximal T rotation]\label{prop:maximal_T_rotation}
    Let $w^{\rm max} \in \mathbb{Z}_2^N$ be a vector obtained by feeding a good nullspace into the maximal rotation subroutine. For any $d \in \mathbb{Z}_2^N$, we denote $T[d] = \bigotimes_{n=1}^N (T_{n, A} \otimes T_{n, B})^{d_n}$. Then, 
    \begin{equation}
        T[w^{\rm max}] U^N_2(t^h, v^h, w^h) T[w^{\rm max}]^{\dag} = U^N_2(t^h, v^h, 0).
    \end{equation}
    Furthermore, let $| \Psi \rangle$ satisfying the Multiple-Squares StateHSP promise, and denote $| \Psi_T \rangle := T[w^{\rm max}] | \Psi \rangle$. Then, 
    \begin{equation}
        U^N_2(t^h, v^h, 0) |\Psi_T \rangle = | \Psi_T \rangle,
    \end{equation}
    and 
    \begin{equation}
        | \langle \Psi_T | U_N^{\otimes 2}(t, v, 0) | \Psi_T \rangle | \leq 1 - \varepsilon,
    \end{equation}
    whenever $(t, v) \neq (t^h, v^h)$. 
\end{prop}

Since the elements $(t, v, 0)$ form the center $Z[\mathcal{D}_4^N]$, which is an abelian subgroup of $\mathcal{D}_4$, what the above proposition tells us is that good null spaces allow us to turn the Multiple-Squares StateHSP into a bona fide \textit{abelian} StateHSP, which is known to be solvable via straightforward Fourier sampling \cite{Hinsche_arXiv_2025_Abelian_stateHSP}. In particular, the operators $U_N^{\otimes 2}(t, v, 0)$ are all doubled Pauli operators, so this problem can now be solved by the Learning Pauli Stabilizers algorithm.

Hence, good null spaces are good. Finally, we also need to show that they are possible, in a finite number of samples.

\bigskip

\textit{Good null spaces are efficiently obtainable }

\begin{prop}[Sampling Complexity of a Good Nullspace]\label{prop:sampling_good_null_space}
    To obtain a good null space from the output of the Bell-resolution subroutine with probability $>1-\delta$, it suffices to sample $M > \varepsilon^{-1} (N + \log \delta^{-1} )$ Bell-resolvable sets.
\end{prop}

The proof for the above relies on the fact that on the projected states $| \Psi(\pi) \rangle$, every element of $\mathcal{D}_4^N$ looks like a Pauli operator, and every bad element (an element that does not belong to any good nullspace) anticommutes with $D^{\pi}(t^h, v^h, w^h)$ for some set of $\pi$'s. This set of $\pi$'s contains enough structure that we may use the anti-concentration in parity samples used to prove \cref{prop:basis_sampling_complexity} to prove the above proposition.

With these set of propositions, we can prove \cref{thm:ms-shsp}. In the following, when we say something is obtained with high probability, we mean that for any fixed $\delta = O(1)$, the probability can be made greater than $1- \delta$ with a number of additional resources that only scale logarithmically in $N$.

\bigskip

\textit{Proof of \cref{thm:ms-shsp}}

\begin{proof}
    If $t^h = 0$, the hidden element will be a Pauli and will be identified in Step 1. If $t^h \neq 0$, we proceed to Step 2, in which we want to be able to create Bell-resolvable sets; according to \cref{prop:basis_sampling_complexity}, each Bell-resolvable set takes $O(N/\varepsilon)$ samples to create. 
    
    The resulting Bell-resolvable state is a representation of $\mathbb{Z}_2^N \times \mathbb{Z}_2^N$, as shown in \cref{prop:conjugacy_ensemble_representation}. In step 3, we perform the Bell-resolution subroutine. According to \cref{prop:Bell-resolution_is_Fourier}, this samples irrep labels $(q, p)$ of $\mathbb{Z}_2^N \times \mathbb{Z}_2^N$, which are guaranteed to satisfy $q \cdot t^h + p \cdot w^h = 0$. We compute the null space $K \leq \mathbb{Z}_2^N \times \mathbb{Z}_2^N$ induced by these irrep labels, which amounts to Gaussian elimination and is thus efficient. 
    
    From \cref{prop:sampling_good_null_space}, after performing this step $O(N/\varepsilon)$ times, with high probability, we will have obtained a $K$ that is a good nullspace, \cref{defn:good_nullspace}. We perform the Maximal Rotation subroutine, and extract a $w^{\rm max}$, which satisfies $t^h \odot w^{\rm max} = w^h$, according to \cref{lem:good_nullspace_output}. 
    
    Subsequent to that, we apply $T[w^{\rm max}]$. According to \cref{prop:maximal_T_rotation}, this maps us exactly into an abelian StateHSP, which is solved with high probability by the Learning Pauli Stabilizers algorithm in $O(N/\varepsilon)$ samples. This correctly identifies $(t^h, v^h)$ with high probability, subsequent to which we can compute $w^h = t^h \odot w^{\rm max}$. 
    
    Since this requires $O(N/\varepsilon)$ Bell-resolvable sets, each of which consume $O(N/\varepsilon)$ copies of the state, subject to the validity of the above propositions, we find that \cref{thm:ms-shsp} holds.
\end{proof}

\section{Proofs}\label{sec:proof_correctness}

In this section, we will prove each of the enabling propositions outlined in the previous section. Readers uninterested in the details (on a first reading or otherwise) may skip to the next section, where we detail how the algorithm generalizes to arbitrary representations.

\subsection{Bell-resolvable sets are efficiently obtainable}

In this section, we will prove \cref{prop:basis_sampling_complexity}. Along the way, we will prove a lemma on the probability masses of subgroups induced by parity measurement, \cref{lem:probability_mass_of_subgroups}, that will come into play again in \cref{ssec:good_nullspace_possible}.

We quote without proof the following lemma, from Refs.~\cite{chen_arXiv_2024_stabilizer_bootstrap, Hinsche_arXiv_2025_Abelian_stateHSP, grewal_arXiv_2024_efficient_learning_states}.
\begin{lem}[Group Sampling Lemma \cite{chen_arXiv_2024_stabilizer_bootstrap, Hinsche_arXiv_2025_Abelian_stateHSP, grewal_arXiv_2024_efficient_learning_states}]\label{lem:group_sampling_lemma}
    Let $\epsilon, \delta \in (0, 1)$ and let $p$ be a distribution over a finite abelian group $G$. Suppose $g_1, ..., g_m$ are i.i.d. samples from $p$, with 
    \begin{equation}
        m \geq \frac{2(\log |G| + \log(1/\delta))}{\epsilon},
    \end{equation}
    then with probability at least $1- \delta$, we have 
    \begin{equation}
        p(\langle g_1, ..., g_m \rangle) \geq 1- \epsilon.
    \end{equation} 
\end{lem}

Our strategy is to identify an upper bound for the probability mass of the the largest proper subgroup of $\mathbb{Z}_2^N$ that can be obtained. 

\begin{lem}[Probability mass of specific subgroups]\label{lem:probability_mass_of_subgroups}
    Let $d \in \mathbb{Z}_2^N$, and let $Z^{\otimes 2}[d]$ be the operator comprising of $Z_{n, A} \otimes Z_{n, B}$ in the spots where $d_n = 1$, and identity otherwise. Suppose the Multiple-Squares StateHSP promise holds, with the further assumption that $t^h \neq 0$, such that $|\langle \Psi | Z^{\otimes 2}[d] | \Psi \rangle | \leq 1 - \varepsilon$ for $d\neq0$, and let $P$ be the probability distribution on $\mathbb{Z}_2^N$ produced by the parity sampling subroutine, \cref{defn:parity_sampling}. Further, let $K \leq \mathbb{Z}_2^N$ be a subgroup of $\mathbb{Z}_2^N$ of dimension $M \leq N$. Then,
    \begin{equation}
        P(K) \leq 1 - \left( 1 - \frac{1}{2^{N-M}} \right) \varepsilon.
    \end{equation} 
\end{lem}

\begin{proof}
    Let $K$ be an $M$ dimensional subgroup of $\mathbb{Z}_2^N$. Then there exists $M$ linearly independent $\pi^1, ..., \pi^M$ that generate the group, in the sense that 
    \begin{equation}
        K = \langle \pi^1, ..., \pi^M \rangle = \left\{ \sum_{m = 1}^M \alpha_m \pi^m | \alpha \in \mathbb{Z}_2^M \right\}.
    \end{equation} 
    We may denote by $[\pi]$ the $M \times N$ matrix formed by setting the entries to be $[\pi]_{m, n} = \pi^m_n$. Now we may write an explicit expression for $P(K)$ as, 
    \begin{widetext}
    \begin{equation}
        \begin{aligned}
            P(K) = P(\langle \pi^1, \pi^2, ... \pi^M \rangle ) 
            &= \sum_{\alpha \in \mathbb{Z}_2^M } P\left(\sum_m \alpha_m \pi^m \right) \\
            &=\frac{1}{2^N} \sum_{\alpha \in \mathbb{Z}_2^M } {\rm Tr} \left[ \prod_{n=1}^N (I + (-1)^{1 +  \sum_m \alpha_m \pi^m_n} Z_{n, A} \otimes Z_{n, B}) | \Psi \rangle \langle \Psi | \right] \\
            &= \frac{1}{2^N} \sum_{\alpha \in \mathbb{Z}_2^M } \sum_{b \in \mathbb{Z}_2^N } (-1)^{\sum_n b_n + \sum_{n, m} \alpha_m b_n \pi_n^m} \langle \Psi | Z^{\otimes 2}[b] | \Psi \rangle  
            \\
            &= \frac{1}{2^N} \sum_{b \in \mathbb{Z}_2^N } (-1)^{\sum_n b_n}  \sum_{\alpha \in \mathbb{Z}_2^M } (-1)^{\sum_n b_n \sum_{m} \alpha_m \pi_n^m} \langle \Psi | Z^{\otimes 2}[b] | \Psi \rangle   .
        \end{aligned}
    \end{equation}
    \end{widetext}
    where
    \begin{align}
        Z^{\otimes 2}[b]
        &=
        \bigotimes_n(Z_{n,A}^{b_n}\otimes Z_{n,B}^{b_n})
    \end{align}
    as in \cref{lem:probability_mass_of_subgroups}.

    We can regard the sum over $\alpha$ as a Fourier transform, and observe that 
    \begin{equation}
        \sum_{\alpha \in \mathbb{Z}_2^M } (-1)^{\sum_{m} \alpha_m \sum_n b_n \pi_n^m}  = 2^M \prod_m \delta_{0, \sum_n b_n \pi_n^m},
    \end{equation}
    and further that the delta function evaluates to $1$ when $b \in {\rm ker} [\pi]$. Since $[\pi]$ has $M$ linearly independent rows, we have that ${\rm dim} ({\rm ker} [\pi]) = N - M$, and so there exist $2^{N-M}$ bitstrings $b$, including $0$, for which the delta function evaluates to $1$. Using this, we have 
    \begin{equation}
        \begin{aligned}
            &\frac{1}{2^N} \sum_{b \in \mathbb{Z}_2^N } (-1)^{\sum_n b_n}  \sum_{\alpha \in \mathbb{Z}_2^M } (-1)^{\sum_n b_n \sum_{m} \alpha_m \pi_n^m} \langle \Psi | Z^{\otimes 2}[b] | \Psi \rangle  \\
            &= \frac{1}{2^{N-M}}  \sum_{b \in \mathbb{Z}_2^N } (-1)^{\sum_n b_n} \prod_m \delta_{0, \sum_n b_n \pi_n^m} \langle \Psi | Z^{\otimes 2}[b] | \Psi \rangle  \\
            &\leq \frac{1}{2^{N-M}}  \sum_{b \in \mathbb{Z}_2^N, b \neq 0} \prod_m \delta_{0, \sum_n b_n \pi_n^m} | \langle \Psi | Z^{\otimes 2}[b] | \Psi \rangle  | + \frac{1}{2^{N-M}} \\
            &= \frac{1}{2^{N-M}}  \sum_{b \in {\rm ker} [\pi], b \neq 0}  | \langle \Psi | Z^{\otimes 2}[b] | \Psi \rangle  | + \frac{1}{2^{N-M}}
        \end{aligned}
    \end{equation}
    where in the inequality, we have simply taken the absolute value of the summands, and separated the term where $b = 0$. Finally, observing that our promise implies
    \begin{equation}
        | \langle \Psi | Z^{\otimes 2}[b] | \Psi \rangle  | \leq 1 - \varepsilon,
    \end{equation} 
    and counting the elements in the kernel, we obtain 
    \begin{equation}
        \begin{aligned}
            &\frac{1}{2^{N-M}}  \sum_{b \in {\rm ker} [\pi], b \neq 0} | \langle \Psi | Z^{\otimes 2}[b] |  + \frac{1}{2^{N-M}} \\ 
            &\leq \frac{1}{2^{N-M}} (1 - \varepsilon) \sum_{b \in {\rm ker} [\pi], b \neq 0} 1 + \frac{1}{2^{N-M}} \\
            &\leq \frac{1}{2^{N-M}} (1 - \varepsilon) (2^{N-M} - 1) + \frac{1}{2^{N-M}}.
        \end{aligned}
    \end{equation}
    Plugging this back into our first line, the result is obtained.
\end{proof}

With this, we can now prove \cref{prop:basis_sampling_complexity}.

\begin{proof}
    By Lagrange's theorem, the largest proper subgroup $K$ of $\mathbb{Z}_2^N$ is of size $2^{N-1}$. Applying \cref{lem:probability_mass_of_subgroups}, we obtain, 
    \begin{equation}
        P(K) \leq 1 - \frac{1}{2} \varepsilon.
    \end{equation}
    Hence, we can apply \cref{lem:group_sampling_lemma} to get that it suffices to sample $O(N/\varepsilon)$ times to obtain, with high probability, a subgroup with probability mass greater than $P(K)$, which hence must be the full group.  
\end{proof}

\subsection{Bell-resolvable sets are useful}

This section proves that Bell resolvable sets combined with the Bell-resolution subroutine implements Fourier sampling on a representation of $\mathbb{Z}_2^N\times\mathbb{Z}_2^N$, which provides the data we need about the conjugacy classes of the reflections on different sites (that is, where the Pauli operators have been conjugated by $T$ gates).

We first prove \cref{prop:conjugacy_ensemble_representation}.

\begin{proof}
    From \cref{eq:parity-rep,eq:parity-ensemble-rep} we have
    \begin{align}
        D^{\{\pi^m\}}
        &=
        \bigotimes_{m,n} D^{\pi_n^m}
        \\
        &=
        \bigotimes_n
        \left[
        \left(\bigotimes_{m:\pi_n^m=0}D^{\pi_n^m}\right)
        \otimes\left(\bigotimes_{m:\pi_n^m=1}D^{\pi_n^m}\right)
        \right]
        \,.
    \end{align}
    When $\pi_n^m=0$ we have the $D^0$ representation of $\mathcal{D}_4$, and by virtue of being a Bell resolvable set, we have an even number of $D^1$ representations for all $\pi_n^m=1$.
    From \cref{eq:conjugacy-representations} we know that on each factor $n$ this gives us a representation of $\mathcal{D}_4/\langle s^2\rangle\cong\mathbb{Z}_2\times\mathbb{Z}_2$, demonstrating that combined over all $N$ sites we have the representation of $\mathbb{Z}_2^N\times\mathbb{Z}_2^N$ claimed.

    Because the parity measurement we perform is projecting onto an eigenstate of an element in the center of the group, the projectors commute with all group elements, including the hidden element, and we have
    \begin{align}
        |\tilde{\Psi}(\{\pi^m\})\rangle
        &=
        \bigotimes_m\Big(\Pi_{\pi^m}|\Psi\rangle\Big)
        \\
        &=
        \bigotimes_m\Big(\Pi_{\pi^m}D(t^h,v^h,w^h)|\Psi\rangle\Big)
        \\
        &=
        D^{\{\pi^m\}}(t^h,v^h,w^h)\bigotimes_m\Big(\Pi_{\pi^m}|\Psi\rangle\Big)
        \\
        &=
        D_\mathcal{C}^{\{\pi^m\}}(t^h,w^h)|\tilde{\Psi}(\{\pi^m\})\rangle
        \,,
    \end{align}
    where for convenience we've used the unnormalized state $|\tilde{\Psi}(\{\pi^m\})\rangle\propto|\Psi(\{\pi^m\})\rangle$.
\end{proof}

We now prove \cref{prop:Bell-resolution_is_Fourier}

\begin{proof}
    The Bell measurements and $X$ measurements project onto the irreps of $D_\mathcal{C}^{\{\pi^m\}}$ as shown in \cref{eq:irrep-bases}.
    This measurement is a fine graining of the character POVM in \cref{lem:character_POVM}, since we project down to a one-dimensional subspace while the character POVM projects onto the full multiplicity space of each irrep, so the statistics of the obtained irrep labels are identical.
    Projecting onto the irreps commutes with the representation itself, so the state being a +1 eigenstate of the hidden element initially means that the hidden element is represented trivially on the irrep we project onto, which is captured by the expression $q\cdot t^h+p\cdot w^h=0$.
\end{proof}

One nice way to interpret the result of Fourier sampling is measuring in the simultaneous eigenbasis of all the elements in the representation $D^{\{\pi^m\}}_\mathcal{C}$, which we formally state below.

\begin{lem}[Measurement outcomes of Bell-resolution subroutine]\label{lem:measurement_subroutine_outcomes}
    Let $(p,q)$ be the outcome of the Bell-resolution subroutine, performed on a Bell-resolvable state labelled by $\{\pi^m\}$. This corresponds to a simultaneous measurement of the operators $D_{\mathcal{C}}^{\{\pi^m\}}(t, w)$ for all $(t, w)$, with the obtained eigenvalue of the $(t, w)$ element given by $(-1)^{q \cdot t^h + p \cdot w^h}$.
\end{lem}

\subsection{Good nullspaces exist, and they are efficiently obtainable}\label{ssec:good_nullspace_possible}

\subsubsection{Preliminaries: Annihilator subgroup formalism}

Having established the analogy between the set of Bell-resolvable states and the abelian StateHSP, we now review the framework set up by \cite{Hinsche_arXiv_2025_Abelian_stateHSP} for solving the abelian StateHSP. Their key observation is that Fourier sampling may be used to solve general abelian StateHSPs, and that one can reason about the outcomes of Fourier sampling via annihilator subgroups.

\begin{defn}[Annhilator subgroup]\label{def:annihilator_subgroup}
    Let $G$ be an abelian group, and $\hat{G}$ its set of characters, and let $H \leq G$ be a subgroup of $G$. The annihilator subgroup is 
    \begin{equation}
        H^{\perp} = \{\chi_{\lambda} \in \hat{G}: \chi_{\lambda}(h) = 1, \forall h \in H \},
    \end{equation}
    which forms a subgroup of $\hat{G}$.
\end{defn}

The outcome of Fourier sampling, i.e. the distribution of irrep labels, is described by the character POVM,
\begin{lem}[Character POVM]\label{lem:character_POVM}
    Fourier sampling on a representation $R$ of an abelian group $G$ yields outcomes $\lambda$ described via the \textit{character POVM}, which has projectors 
    \begin{equation}
        \Pi_{\lambda} = \frac{1}{|G|} \sum_{g \in G} \overline{\chi_{\lambda}(g)} R(g).
    \end{equation}
    When acting on a state $| \Psi \rangle$, the induced probability distribution is 
    \begin{equation}
        q(\lambda | \Psi ) = \frac{1}{|G|} \sum_{g \in G} \overline{\chi_{\lambda}(g)} \langle \Psi | R(g) | \Psi \rangle.
    \end{equation}
\end{lem}

Since the annihilator subgroup of $K$ is defined to be the set of $\lambda$ such that $\chi_{\lambda}(k) = 1$ for all $k \in K$, the following lemma (proven in \cite{Hinsche_arXiv_2025_Abelian_stateHSP}) follows: 
\begin{cor}[Character POVM on subgroups]\label{cor:character_POVM_subgroups}
    Let $G, R, |\psi \rangle$ be as in \cref{lem:character_POVM}. Let $K \leq G$ be a subgroup of $G$, and $K^{\perp}$ be its annihilator. Then,
    \begin{equation}
        q(K^{\perp} | \Psi ) = \frac{1}{|K|} \sum_{g \in K} \langle \Psi | R(g) | \Psi \rangle.
    \end{equation}
\end{cor}

The main idea is that Fourier sampling, via the character POVM, establishes a probability distribution on $q$ on the group of characters, such that Fourier sampling may be thought about as obtaining i.i.d. samples from this distribution. Due to the group structure of the characters, one may use the group sampling lemma, \cref{lem:group_sampling_lemma} to ensure that the span of the sampled characters have a minimum probability mass. The hidden subgroup $H$ of an abelian StateHSP is associated with an annihilator subgroup $H^{\perp}$ that by the problem definition must have probability mass $1$, and if one can bound the probability mass of the largest subgroup $K^{\perp} < H^{\perp}$ to be $\varepsilon$ away from $1$, then one can ensure that the full $H^{\perp}$ may be sampled with high probability over $O(\log |G| / \varepsilon)$ samples.

Since the assumptions of the abelian StateHSP may be violated for any particular instance of a Bell-resolvable state, we cannot directly apply this strategy. However, as we will soon see, we can adapt this strategy by allowing ourselves some leeway. We will not have to sample a particular annhilator subgroup $K^{\perp}$ with high probability -- instead there are a large class of annihilator subgroups which would equivalently solve our problem; we just have to ensure that we obtain at least one of them.

\subsubsection{Bounding the probability mass of annihilator subgroups}

Since the Bell-resolution subroutine is performing Fourier sampling on a representation of $G = \mathbb{Z}_2^N \times \mathbb{Z}_2^N$, we will want to build a probability mass $q: \hat{G} \rightarrow [0, 1]$ on annihilator  subgroups.
Elements of the annihilator are characters of $\mathbb{Z}_2^N\times\mathbb{Z}_2^N$, which since all the irreps are one-dimensional can simply be thought of as the irreps, which as in \cref{eq:z2z2-irreps} are labeled by $(p,q)$, so we can write
\begin{equation}
    \chi_{q, p}(t, w) = e^{i \pi (q \cdot t + p \cdot w)}.
\end{equation}
We can then form annhilator subgroups by taking the span of outcomes of the Bell-resolution subroutine. Having obtained the annihilator subgroup $K^{\perp}$ after performing some number of Bell-resolution subroutines, we are guaranteed that $(t^h, w^h) \in K$. 

A natural thing for the algorithm to do then would be to sample until $K$ becomes one-dimensional; $(t^h, w^h)$ must then be the unique vector in the null space. However, it is not obvious how many samples it would take to ensure this with high probability.

In the framework of \cite{Hinsche_arXiv_2025_Abelian_stateHSP}, the authors built a probability measure for annihilator subgroups using the original state fulfilling the abelian StateHSP promise. In contrast, in every round of our Bell-resolution subroutine, we obtain a different representation labelled by $\{\pi^m\}$. As such, the character POVM takes a conditional form, as exemplified in the following lemma. 

\begin{lem}[Minor bound on annihilator mass]\label{lem:minor_bound_annihilator_mass}
    Let $q$ be the probability distribution on outcomes of the Bell-resolution subroutine, averaged over Bell-resolvable states. Then for any annihilator subgroup $K^{\perp} \leq \mathbb{Z}_2^N \times \mathbb{Z}_2^N$, the following bound holds:
    \begin{equation}
        q(K^{\perp}) \leq \frac{1}{|K|}  \sum_{(t, w) \in K} \sum_{\pi} p(\pi) | \langle \Psi(\pi) | D_{\mathcal{C}}^{\pi}(t, w) | \Psi(\pi) \rangle |.
    \end{equation}
\end{lem}

\begin{proof}
    From \cref{lem:character_POVM},  a Bell-resolvable set $\{\pi^m\}_{m=0}^M$ induces the distribution, 
    \begin{equation}
      \begin{aligned}
        q(K^{\perp} | \{\pi^m\}) = \frac{1}{|K|} \sum_{(t, w) \in K} \prod_{m=0}^M \langle \Psi(\pi^m) | D_{\mathcal{C}}^{\pi^m}(t, w) | \Psi(\pi^m) \rangle 
      \end{aligned}
    \end{equation}
    The unconditional distribution is then given by summing over the conditional part. The key insight is that the first entry $\pi^0$ is drawn from $| \Psi \rangle$ with distribution $p(\pi^0)$. However, the remaining distribution $p(\{ \pi^{m} \}_{m=1}^M | \pi^0)$ depends on $\pi^0$ in a complicated way that we will not have to specify. In the following, $p$ refers to several distinct probability distributions, but which one will be clear from its argument.
\begin{equation}
  \begin{aligned}
    q(K^{\perp}) 
    &= \sum_{\{\pi^m\}_{m=0}^M} q(K^{\perp} | \{\pi^m\}) p(\{\pi^m\}) \\
    &= \sum_{\pi^0} \sum_{\{\pi^m\}_{m=1}^M} q(K^{\perp} | \{\pi^m\}) p(\pi^0) p(\{ \pi^{m} \}_{m=1}^M | \pi^0)
  \end{aligned}
\end{equation}
Plugging in the prior expression and rearranging terms, 
\begin{widetext}
\begin{equation}
  \begin{aligned}
    &q(K^{\perp}) \\
    &= \sum_{\pi^0} \sum_{\{\pi^m\}_{m=1}^M} q(K^{\perp} | \{\pi^m\}) p(\pi^0) p(\{ \pi^{m} \}_{m=1}^M | \pi^0) \\
    &= \sum_{\pi^0} \sum_{\{\pi^m\}_{m=1}^M} p(\pi^0) p(\{ \pi^{m} \}_{m=1}^M | \pi^0)  \frac{1}{|K|} \sum_{(t, w) \in K} \prod_{m=0}^M  \langle \Psi(\pi^m) | D_{\mathcal{C}}^{\pi^m}(t, w) | \Psi(\pi^m) \rangle  \\
    &= \frac{1}{|K|} \sum_{\pi^0} p(\pi^0)\langle \Psi(\pi^0) | D_{\mathcal{C}}^{\pi^0}(t, w) | \Psi(\pi^0) \rangle \sum_{\{\pi^m\}_{m=1}^M}  p(\{ \pi^{m} \}_{m=1}^M | \pi^0) \prod_{m=0}^M \langle \Psi(\pi^m) | D_{\mathcal{C}}^{\pi^m}(t, w) | \Psi(\pi^m) \rangle \\
    &\leq  \frac{1}{|K|}  \sum_{(t, w) \in K} \sum_{\pi^0} p(\pi^0) \left|\langle \Psi(\pi^0) | D_{\mathcal{C}}^{\pi^0}(t, w) | \Psi(\pi^0) \rangle \right| \left| \sum_{\{\pi^m\}_{m=1}^M}  p(\{ \pi^{m} \}_{m=1}^M | \pi^0) \prod_{m=0}^M \langle \Psi(\pi^m) | D_{\mathcal{C}}^{\pi^m}(t, w) | \Psi(\pi^m) \rangle \right| \\
    &\leq \frac{1}{|K|}  \sum_{(t, w) \in K} \sum_{\pi^0} p(\pi^0) \left| \langle \Psi(\pi^0) | D_{\mathcal{C}}^{\pi^0}(t, w) | \Psi(\pi^0) \rangle\right|,
  \end{aligned}
\end{equation}
\end{widetext}
where in the first inequality we have just used that the sum is bounded by the sum of absolute values, and in the second inequality we have used that the second term is some distribution over pauli expectation values and so must be bounded by $1$. In the last line, the term in the absolute sign is also an expectation value of a Pauli operator and is bounded by $1$. 
\end{proof}

The form of the bound in \cref{lem:minor_bound_annihilator_mass} reveals why there are difficulties with getting sample complexity guarantees for obtaining a unique null space -- a state can satisfy $|\langle \Psi | U^{\otimes 2}_N(t, 0, w) | \Psi \rangle | \leq 1 - \varepsilon$, yet still satisfy $|\langle \Psi(\pi) | D^\pi_{\mathcal{C}}(t, w) | \Psi(\pi \rangle | = 1$ for all $\pi$! To see this, we may consider a state $| \Psi \rangle$ built from randomly allocating $| \Psi(\pi) \rangle$ to be $\pm 1$ eigenstates of $D^\pi_{\mathcal{C}}(t, w)$. Then $|\langle \Psi(\pi) | D^\pi_{\mathcal{C}}(t, w) | \Psi(\pi) \rangle | = 1$ for all $\pi$, but $|\langle \Psi | U^{\otimes 2}_N(t, 0, w) | \Psi \rangle |$ will be close to $0$ \footnote{We note that this presents an obstruction for using \cref{lem:minor_bound_annihilator_mass} to prove our result, although it may well be that one can obtain alternative bounds using a proper accounting of the conditional distribution $p(\{\pi^m\} | \pi^0)$. We do not attempt to do so in this work.}.

Fortunately for us, the algorithm does not need to obtain a unique null space to solve the problem, and there is a larger class of good null spaces which still admit a solution of the problem. In the next section, we define such good null spaces. Then, we will compute the above bound for some bad null spaces, and use it to argue that we can obtain with high probability a good null space.

\subsubsection{What is a good null space?}

At first glance, it seems like we require the Bell-resolution subroutine to yield a one-dimensional null space to learn the unique $(t^h, w^h)$. However, this seems hard to guarantee.
We sidestep this potential problem with the insight that we don't need to learn $(t^h, w^h)$.
Instead, we don't need to know $t^h$ at all, and we only need to learn some $w$ that agrees with $w^h$ on the support of $t^h$.
Applying $T^{\dag}$ to our state everywhere this $w$ is $1$ will then suffice to rotate the T-conjugated stabilizer into a Pauli stabilizer.

This motivates the Maximal Rotation subroutine, recalled here:

\begin{defn*}[Maximal rotation subroutine]
    The maximal rotation subroutine takes as input a vector space $K$ over $\mathbb{Z}_2^N \times \mathbb{Z}_2^N$, and outputs a vector $w^{\rm max} \in \mathbb{Z}_2^N$ via the following procedure:
    \begin{enumerate}
        \item Obtain a basis $(t^n, w^n)$ for $K$.
        \item Obtain $w^{\rm max}$ by performing bit-wise OR between all $w^n$.
    \end{enumerate}
\end{defn*}

When does this procedure work? We recall the definition and lemma associated with a good nullspace. 

\begin{defn*}[Good nullspace]
    Let $K \subseteq \mathbb{Z}_2^N \times \mathbb{Z}_2^N$ be a subgroup containing $(t^h, w^h)$. Then we say that $K$ is a \textit{good nullspace} if for all elements $(t, w)$ we have,
    \begin{equation}
      (t^h \odot w)_j = 1 \implies w^h_j = 1.
    \end{equation}
\end{defn*}

The condition above tells us that most vectors in a good nullspace $K$ will not match $(t^h, w^h)$. However, all that we require is that in spots where (the unknown) $t^h_j = 1$, the nullspace $K$ should not get $w$ wrong. In Step 2b. of the algorithm, the nullspace obtained from the previous step is fed into the maximal rotation subroutine. The output is characterized as follows.

\begin{lem*}[Output of a good nullspace]
    Let $K$ be a good nullspace. Given such a group as input into the Maximal Rotation subroutine, the output will be of the form $w^{\rm max} = w^h + (t^h)^c \odot w' \in \mathbb{Z}_2^N$, where here $(t^h)^c$ denotes the complement of $t^h$, and the form of $w'$ is unimportant.
\end{lem*}

\begin{proof}
    For every spot in the support of $t^h$ where $w^h_n = 1$, there must exist at least one basis element where $w^h_n = 1$, otherwise we will not be able to generate $(t^h, w^h)$. Conversely, for every spot in the support of $t^h$ where $w^h_n = 0$, by assumption we have $w_n = 0$ for the entire basis as well. Hence $w^{\rm max} = w^h + (t^h)^c \odot w'$.
\end{proof}

To summarize, given a subgroup satisfying the conditions of \cref{defn:good_nullspace}, we may efficiently and deterministically extract a vector $w^{\rm max}$ of the form $w^h + (t^h)^c \odot w'$, where the form of $w'$ does not matter.

Finally, we argue that learning such a $w$ suffices to reduce our problem to an abelian StateHSP, in particular the Learning Pauli Stabilizers problem, \cref{def:pauli_sampling}.

\begin{lem*}[Maximal T rotation]
    Let $w^{\rm max} \in \mathbb{Z}_2^N$ be such that $t^h \odot w^{\rm max} = w^h$, and $w' = (t^h)^c \odot w^{\rm max}$ is an arbitrary vector. Further, for any $d \in \mathbb{Z}_2^N$, let $T[d] = \bigotimes_{n=1}^N (T_{n, A} \otimes T_{n, B})^{d_n}$. Then, 
    \begin{equation}
        T[w^{\rm max}] U^{\otimes 2}_N(t^h, v^h, w^h) T[w^{\rm max}]^{\dag} = U^{\otimes 2}_N(t^h, v^h, 0).
    \end{equation}
    As a consequence, for $| \Psi \rangle$ satisfying the Multiple-Squares StateHSP promise, we then have 
    \begin{equation}
        \begin{aligned}
            U^{\otimes 2}_N(t^h, v^h, 0) T[w^{\rm max}] | \Psi \rangle = T[w^{\rm max}] | \Psi \rangle
        \end{aligned}
    \end{equation}
\end{lem*}

\begin{proof}
    We may proceed by direct computation. Recall that 
    \begin{equation}
        T = e^{i \pi Z / 8}.
    \end{equation}
    First, consider a site where $t^h_n = 0$. In that case, by assumption $w^h_n = 0$ (otherwise the hidden element will not be an involution). The relevant operator $U_2(0, v^h_n, 0) = (Z_n \otimes Z_n)^{v^h_n}$ is a $Z$-type operator and hence commutes with $T$. Thus,
    \begin{equation}
        \begin{aligned}
            &(T_{n, A} \otimes T_{n, B})^{w^{\rm max}_n} U_2(0, v^h_n, 0) (T_{n, A}^{\dag} \otimes T_{n, B}^{\dag})^{w^{\rm max}_n} \\&= U_2(0, v^h_n, 0).
        \end{aligned}
    \end{equation}
    Next consider the sites where $t^h_n = 1$. In that case, $w^{\rm max}_n = w^h_n$, and 
    \begin{equation}
    \begin{aligned}
        &U_2(1, v^h_n, w^h_n) 
        = \\&X_{n, A} e^{i \pi(2 v^h_n + w^h_n) Z_{n, A}/4} \otimes X_{n, B} e^{i \pi(2 v^h_n + w^h_n) Z_{n, B}/4}
    \end{aligned}
    \end{equation}
    contains an $X$ operator, so 
    \begin{equation}
    \begin{aligned}
        &(T_{n, A} \otimes T_{n, B})^{w^h_n}  U_2(1, v^h_n, w^h_n) \\
        &=   U_2(1, v^h_n, w^h_n)(T_{n, A}^{\dag} \otimes T_{n, B}^{\dag})^{w^h_n}
    \end{aligned}
    \end{equation}
    Finally, we have 
    \begin{equation}
    \begin{aligned}
        &(T_{n, A}^{\dag} \otimes T_{n, B}^{\dag})^{w^h_n} (T_{n, A}^{\dag} \otimes T_{n, B}^{\dag})^{w^h_n} \\
        &=  e^{-i \pi w^h_n Z_{n, A}/4} \otimes e^{-i \pi w^h_n Z_{n, B}/4},
    \end{aligned}
    \end{equation}
    hence, 
    \begin{equation}
    \begin{aligned}
        &(T_{n, A} \otimes T_{n, B})^{w^h_n}  U_2(1, v^h_n, w^h_n)  (T_{n, A}^{\dag} \otimes T_{n, B}^{\dag})^{w^h_n} \\
        &=  U_2(1, v^h_n, 0). 
    \end{aligned}
    \end{equation}
    Combining these facts for all the locations yields the result.
\end{proof}

\begin{prop}[Maximal T rotation gives an abelian StateHSP]\label{prop:maximal_T_rotation_abelian_stateHSP}
    Let $|\Psi \rangle$ be a state satisfying the Multiple-Squares StateHSP promise, and $w^{\rm max} \in \mathbb{Z}_2^N$ be such that $t^h \odot w^{\rm max} =  w^h$, and $w' = (t^h)^c \odot w^{\rm max}$ is an arbitrary vector. The state $|\Psi_T \rangle = T[w^{\rm max}]^{\dag} | \Psi \rangle$ satisfies 
    \begin{equation}
        U^{\otimes 2}_N(t^h, v^h, 0) |\Psi_T \rangle = | \Psi_T \rangle,
    \end{equation}
    and 
    \begin{equation}
        | \langle \Psi_T | U^{\otimes 2}_N(t, v, 0) | \Psi_T \rangle | \leq 1 - \varepsilon,
    \end{equation}
    whenever $(t, v) \neq (t^h, v^h)$. Now, $U^{\otimes 2}_N(t, v, 0)$ is a representation of $(t, v) \in \mathbb{Z}_2^N \times \mathbb{Z}_2^N$, so this is an abelian StateHSP.
\end{prop}
\begin{proof}
    The first part of the proposition was proved in the previous lemma. The second part follows because for any $(t, v) \neq (t^h, v^h)$, there exists some $(t, v, w')$ such that $T[w^{\rm max}] \rho(t, v, w') T[w^{\rm max}]^{\dag} = \rho(t, v, 0)$, and $| \langle \Psi_T | U^{\otimes 2}_N(t, v, 0) | \Psi_T \rangle | \leq 1 - \varepsilon$ follows from the Multiple-Squares StateHSP promise. Finally, we may notice that $U^{\otimes 2}_N(t, v, 0)$ are all two-fold tensor products of Pauli operators, which hence form a representation of $\mathbb{Z}_2^N \times \mathbb{Z}_2^N$. 
\end{proof}

In fact the remaining abelian StateHSP is exactly the learning Pauli stabilizers problem.

\subsubsection{A good null space is efficiently obtainable}

We have defined good null spaces, and argued that they suffice to solve the problem. We now have to show that it is possible to obtain good null spaces in a reasonable number of samples. 

To begin with, what is the minimal bad subgroup $K_b$ that causes the Maximal Rotation subroutine to fail? This means $K_b$ must contain some element of the form $(t', w')$ where $w' = w^b \oplus w^r$, where $b$ stands for bad and $r$ stands for rest, and $t^h \odot w^b$ is $1$ in at least one spot where $w^h = 0$. Since this means $(t', w') \neq (t^h, w^h)$, $K_b$ contains at least $4$ elements, 
\begin{equation}
  \begin{aligned}
    (0, 0), (t^h, w^h), (t', w^b + w^r), (t' + t^h, (w^h + w^b) + w^r),
  \end{aligned}
\end{equation}
where $w^b = t^h \odot w^b$ is only supported on $t^h$, and $w^r = (t^h)^c \odot w^r$ is only supported on the complement. The next proposition tells us that the probability mass of the annihilators $K_b^{\perp}$ of such bad subgroups is bounded away from $1$.

\begin{prop}[Bad null spaces are concentrated]\label{prop:bad_null_spaces_unlikely}
    Let $K$ be a \textit{bad null space}, namely a subgroup that, when fed as input, will cause the Maximal Rotation subroutine to output a bad $w$. Then, the probability mass of $K^{\perp}$ obtained from the Bell resolution subroutines, for a state $| \Psi \rangle$ satisfying the Multiple-Squares StateHSP promise, with the additional assumption that $t^h \neq 0$, is bounded as
    \begin{equation}
        q(K^{\perp}) \leq 1 - \frac{1}{2} \varepsilon.
    \end{equation}
\end{prop}

Before embarking on the proof, we need a couple small lemmas characterizing when and where different representatives commute.

\begin{lem}[Commutator of representatives]\label{lem:repn_cmm_form}
    Let $(t, w), (t', w') \in \mathbb{Z}_2^N \times \mathbb{Z}_2^N$. Then 

    \begin{equation}
    \begin{aligned}
    \left[ D^{\pi}_{\mathcal{C}}(t, w), D^{\pi}_{\mathcal{C}}(t^\prime, w^\prime) \right] = 0 
    &\iff 0 = \sum_n \pi_n (t_n w_n' + t_n' w_n),
    \end{aligned}
    \end{equation}
\end{lem}

\begin{proof}
    This follows from expanding the representations explicitly to write 
    \begin{equation}
  \begin{aligned}
    \left[ D^{\pi}_{\mathcal{C}}(t^h, w^h), D^{\pi}_{\mathcal{C}}(t, w)\right] &\propto
     \prod_{n=1}^N \left[ X^{t_n} Z^{w_n \pi_n}, X^{t_n'} Z^{w_n' \pi_n} \right],
      \end{aligned}
    \end{equation}
    and the lemma follows as a standard exercise in Pauli algebra.
\end{proof}

\begin{lem}[Bad elements must anticommute somewhere]\label{lem:bad_elt_anticommute}
 Let $(t, w)$ be such that for some $J$, $w_n^h = 1$ but $(t^h \odot w)_n = 0$. Then there exists some $\pi^J$ for which 
 \begin{equation}
     \left\{ D^{\pi^J}_{\mathcal{C}}(t^h, w^h), D^{\pi^J}_{\mathcal{C}}(t, w)\right\} = 0.
 \end{equation}
\end{lem}
\begin{proof}
    Set $\pi^J$ to be $1$ in the $J$-th position and $0$ everywhere else. The result then follows from \cref{lem:repn_cmm_form}.
\end{proof}

\begin{lem}[Bad elements commute on a group]\label{lem:bad_elt_commute}
 Consider $(t, w), (t', w') \in \mathbb{Z}_2^N \times \mathbb{Z}_2^N$, and let $S_{\rm cmm}[t, w]$ be the set of $\pi \in \mathbb{Z}_2^N$ such that 
 \begin{equation}
     \left[ D^{\pi}_{\mathcal{C}}(t, w), D^{\pi}_{\mathcal{C}}(t', w')\right] = 0.
 \end{equation}
 Then $S_{\rm cmm}$ is a subgroup of $\mathbb{Z}_2^N$.
\end{lem}
\begin{proof}
    This follows from \cref{lem:repn_cmm_form}. Explicitly: $\sum_n \pi^1 (t_n w_n' + t_n' w_n) = 0$ and $\sum_n \pi^2 (t_n w_n' + t_n' w_n) = 0$, implies $\sum_n (\pi^1+\pi^2) (t_n w_n' + t_n' w_n) = 0$, so $S_{\rm cmm}$ is a subgroup. 
\end{proof}

With these simple set of properties of bad elements in mind, we can now prove \cref{prop:bad_null_spaces_unlikely}. 

\begin{proof}
    Suppose $K_b$ be a minimal bad null space, in the sense that any bad null space $K$ must contain as a subgroup some such $K_b$.
    
    Recall that this means $K_b$ must contain some element of the form $(t^b, w^b)$, and there exists $J$ for which $t^h_J w^b_J = 1$ but $w^h_J = 0$. Since this means $(t', w') \neq (t^h, w^h)$, $K_b$ contains at least $4$ elements, 
    \begin{equation}
      \begin{aligned}
        (0, 0), (t^h, w^h), (t', w^b + w^r), (t' + t^h, (w^h + w^b) + w^r).
      \end{aligned}
    \end{equation}
    
    Applying \cref{lem:minor_bound_annihilator_mass}, the probability mass of the annihilator of this subgroup is bounded by
\begin{equation}
  \begin{aligned}
    &q(K^{\perp}_b) \\
    &\leq \frac{1}{4} \sum_{(t, w) \in K_b} \sum_{\pi} p(\pi) |  \langle \Psi(\pi) | D_{\mathcal{C}}^{\pi}(t, w)  |\Psi(\pi) \rangle| \\
    &= \frac{1}{2} + \frac{1}{4} \sum_{\pi}  p(\pi) |  \langle \Psi(\pi) | D_{\mathcal{C}}^{\pi}(t^b, w^b)  |\Psi(\pi) \rangle|  \\
    &\qquad + \frac{1}{4} \sum_{\pi}  p(\pi) |  \langle \Psi(\pi) | D_{\mathcal{C}}^{\pi}(t^h +t^b, w^h + w^b)  |\Psi(\pi) \rangle| .
  \end{aligned}
\end{equation}
Our strategy with dealing with these terms, for given $(t, w)$, is to split $\pi$ into two sets: 
\begin{equation}
  \begin{aligned}
    S_{\rm cmm}[t, w]: &\pi \in S_{\rm cmm}[t, w] \\
    &\qquad \implies \left[ D_{\mathcal{C}}^{\pi}(t^h, w^h), D_{\mathcal{C}}^{\pi}(t, w) \right] = 0, \\
    S_{\rm acmm}[t, w]: &\pi \in S_{\rm acmm}[t, w]  \\
    &\qquad \implies \left\{ D_{\mathcal{C}}^{\pi}(t^h, w^h), D_{\mathcal{C}}^{\pi}(t, w) \right\} = 0.
  \end{aligned}
\end{equation}
Since all $D_{\mathcal{C}}^{\pi}$ are paulis, we must have \begin{equation}
  \begin{aligned}
    S_{\rm cmm}[t, w] \cup S_{\rm acmm}[t, w] &= \mathbb{Z}_2^N, 
    \\  S_{\rm cmm}[t, w] \cap S_{\rm acmm}[t, w] &= \emptyset.
  \end{aligned}
\end{equation}
Then, given that $D_{\mathcal{C}}^{\pi}(t^h, w^h) | \Psi(\pi) \rangle = | \Psi(\pi) \rangle$, a sum of the above form then splits and is bounded like,
\begin{equation}
  \begin{aligned}
    & \sum_{\pi}  p(\pi) |  \langle \Psi(\pi) |   D_{\mathcal{C}}^{\pi}(t, w) |\Psi(\pi) \rangle | \\
    &= \sum_{\pi \in S_{\rm cmm}[t, w]}  p(\pi) |  \langle \Psi(\pi) |   D_{\mathcal{C}}^{\pi}(t, w) |\Psi(\pi) \rangle | \\
    &\qquad + \sum_{\pi \in S_{\rm acmm}[t, w]}  p(\pi) |  \langle \Psi(\pi) |   D_{\mathcal{C}}^{\pi}(t, w) |\Psi(\pi) \rangle | \\
    &=  \sum_{\pi \in S_{\rm cmm}[t, w]}  p(\pi) |  \langle \Psi(\pi) |   D_{\mathcal{C}}^{\pi}(t, w) |\Psi(\pi) \rangle | \\
    &\leq  \sum_{\pi \in S_{\rm cmm}[t, w]}  p(\pi).
  \end{aligned}
\end{equation}

Now consider $(t, w) = (t^b, w^b), (t^h + t^b, w^h + w^b)$. Note that both are bad elements, because $w^h_J+w^b_J=1$ when $w^b_J=1$ and $w^h_J=0$. Hence, by \cref{lem:bad_elt_commute}, $S_{\rm cmm}$ forms a group, and by \cref{lem:bad_elt_anticommute}, $S_{\rm acmm}$ is non-empty, so $S_{\rm cmm}$ must be a proper subgroup.

Applying \cref{lem:probability_mass_of_subgroups}, we get that the probability masses of $S_{\rm cmm} \leq 1 - \varepsilon/2$. This implies 
\begin{equation}
\begin{aligned}
    q(K_b^{\perp}) &\leq \frac{1}{2} + \frac{1}{4} \sum_{\pi \in S_{\rm cmm}[t^b, w^b]} p(\pi) \\
    &\qquad + \frac{1}{4} \sum_{\pi \in S_{\rm cmm}[t^b+t^h, w^b+w^h]} p(\pi) \leq 1 - \frac{1}{2} \varepsilon.
\end{aligned}
\end{equation}
We are not quite there yet. As a final step, let us consider a bad subgroup $K$. Then there must exist some minimal bad subgroup $K_b \subseteq K$. This implies $K^{\perp} \subseteq K_b^{\perp}$. Since $q$ is a probability measure, we must have 
\begin{equation}
  q(K^{\perp}) \leq q(K_b^{\perp}) \leq  1 - \frac{1}{2} \varepsilon.
\end{equation}
\end{proof}

The sampling complexity of the Maximal Rotation subroutine follows.
\begin{cor}{Sampling complexity of a good null space}\label{cor:sampling_good_null_space}
    To obtain a good null space from the output of the Bell-resolution subroutine with probability $1-\delta$, it suffices to sample $O(N/\varepsilon)$ times.
\end{cor}

\begin{proof}
    This follows from \cref{prop:bad_null_spaces_unlikely} and the group sampling lemma, \cref{lem:group_sampling_lemma}.
\end{proof}

\section{Extension to Arbitrary Representations}\label{sec:arbitrary_reps}

Having solved the problem for $U_2 = D^0 \oplus D^1$, we now proceed to show how this problem generalizes to the StateHSP for an arbitrary representation, $R$. The primary resource we must assume is access to an efficient partial Fourier transform that block-diagonalizes within the multiplicity space. This is an operation $U_{\rm pFT}$ that acts like,
\begin{equation}\label{eq:complete_pFT}
\begin{aligned}
    &U_{\rm pFT}^{\dag} R(t, v, w) U_{\rm pFT} = \sum_{\sigma, \pi \in \mathbb{Z}_2} \sum_{\alpha = 0}^{M_{\sigma, \pi}} \\
    &\qquad | \alpha \rangle \langle \alpha |_A \otimes | \sigma \rangle \langle \sigma |_S \otimes | \pi \rangle \langle \pi |_P \otimes D^{2 \sigma + \pi}(t, v, w),
\end{aligned}
\end{equation}
where $\alpha$ is a multiplicity register, $M_{\sigma, \pi}$ is the multiplicity of the representation $D^{2 \sigma + \pi}$, and $\sigma, \pi$ together label the reps that appear as defined in \cref{eq:dihedral-2d-reps}. For easy reference, we label the registers $A, S, P, D$. There is some arbitrariness in this choice of Fourier transform, for which we provide a couple of remarks.

First, there is some ambiguity in defining $M_{\sigma, \pi}$, since $D^{1}, D^3$ are equivalent and may be labeled either way. We split them arbitrarily simply to make our calculations more symmetric in $\sigma$. All registers are qubit registers except for the multiplicity register, which has an arbitrary dimension depending on the details of the representation. One can always choose this register to match the size of the maximum $M_{\sigma, \pi}$.

Along the same lines, doing a full Fourier transform over $\mathcal{D}_4$ will do as well. The crucial part of this transformation is that it should block-diagonalize even within the multiplicity space, which is a representation-dependent detail. Ultimately, our convention for the Fourier transform in \cref{eq:complete_pFT} was chosen for the purposes of notational simplicity, as well as continuity with the previous algorithm -- note that this is exactly the transformation carried out by the parity projections in \cref{eq:parity-projections}. 

As previously noted, the HSP is simply the StateHSP on the regular representation with $\varepsilon = 1$, solving this will also solve the HSP on $\mathcal{D}_4^N$ \cite{Bouland_arXiv_2024_stateHSP}. On the regular representation, \cref{eq:complete_pFT} has a convenient interpretation as a conditional Fourier transform. In particular, one may encode the group element $r^t s^k$ in a qubit and ququart register as $| t \rangle | k \rangle$. The transformation \cref{eq:complete_pFT} is then carried out by a Fourier transform of the second register, with the sign of the momentum conditional on the first \cite{Bacon_arXiv_2008_dihedral_simon}.

At first glance, one might suspect that the inequivalence of the irreps appearing in $D^0$, $D^2$ obstruct our algorithm. However, we recall from the discussion in \cref{sec:dihedral_prelims} that these are both reducible to 1D irreps, and admit Fourier sampling via a single qubit measurement. Hence, as long as we can pair the 2D irreps, which are equivalent, we can still form the conjugacy ensemble representation, which will still be a representation of $\mathbb{Z}_2^N \times \mathbb{Z}_2^N$. The only difference lies in processing the outcomes of the Bell-resolution step to obtain the irreps, where one has reinterpret the signs of some measurement outcomes.

The algorithm for this generalized Multiple-Squares StateHSP proceeds almost identically to before, subject to several modifications. 
\begin{enumerate}
    \item Instead of performing parity sampling, we need to perform what we wll refer to as partial Fourier sampling, namely we will measure the labels $\alpha, \sigma, \pi$. We refer to this as `partial' Fourier sampling as we will not need to measure the 1D irreps at this stage. We note this is similar to the sampling used to obtain coset states in the HSP version of this problem \cite{Bacon_arXiv_2008_dihedral_simon, Kuperberg_arXiv_2004_sieve, childs_quantum_2010}.
    \item The states we obtain from partial Fourier sampling now have three labels $\alpha, \sigma, \pi$. We will form Bell-resolvable sets in the same way, by requiring that $\sum_{m=0}^M \pi^m = 0$. However, the details of the representation will also depend on $\sigma$, and the (irrelevant) details of the sampled state further depend on $\alpha$. The representation is still of $\mathbb{Z}_2^N \times \mathbb{Z}_2^N$, however, the classical processing of the measurement outcomes of the  Bell-resolution subroutine must be modified accordingly.
    \item When we obtain $w^{\rm max}$, instead of applying $T$ gates, we will apply controlled-S operations in positions dictated by $w^{\rm max}$, conjugated by $U_{\rm pFT}$.
    \item We will still use the Learning Pauli Stabilizers subroutine with the small modification that we are now learning a representation of the phaseless Pauli group $\mathbb{Z}_2^N\times\mathbb{Z}_2^N<\mathcal{D}_4^N$ embedded in $R^N$ rather than actual Pauli stabilizers. Due to \cref{fact:abelian_stateHSP}, this is no harder than the original Pauli learning problem.
\end{enumerate}

We now formally state the generalized problem.

\begin{defn}[General Multiple-Squares StateHSP]\label{def:Dihedral_StateHSP_II}
    Let $| \Psi \rangle$ be a state on $N$ qudits, and let $R^N$ denote some tensor product representation of $\mathcal{D}_4^{N}$ (so that $R$ is a representation of $\mathcal{D}_4$).

    We are promised that $| \Psi \rangle$ hides an involution, which is a subgroup of the form $\mathcal{H} = \{(0, 0), (t^h, v^h, w^h)\} < \mathcal{D}_4^{N}$, where $t^h \neq 0$, such that the state satisfies 
    \begin{gather*} 
        R^N(t^h, v^h, w^h)  | \Psi \rangle = | \Psi \rangle, \\
        | \langle \Psi | R^N(t, v, w) | \Psi \rangle | \leq 1 - \varepsilon, \mbox{ for all } (t, v, w) \notin\mathcal{H}. 
    \end{gather*}   
    Given access to copies of $|\Psi \rangle$, the problem is to identify $\mathcal{H}$, or equivalently $(t^h, v^h, w^h)$. We assume access to the complete partial Fourier Transform identified in \cref{eq:complete_pFT}, as well as controlled-$S$ gates.
\end{defn}

We state the main result of this section as a theorem.

\begin{thm}[General Multiple-Squares StateHSP]\label{thm:Dihedral_StateHSP_theorem_II}
  There exists an algorithm that solves the General Multiple-Squares StateHSP problem (\cref{def:Dihedral_StateHSP_II}) with high probability using $O(N^2/\varepsilon^2)$ samples.
\end{thm}

As a corollary, we obtain the traditional HSP.
\begin{cor}
    The same algorithm solves the Multiple-Squares HSP on the regular representation using $O(N^2)$ samples.
\end{cor}

In the next section, we will detail systematically the modified algorithm. Subsequently we will address the modified subroutines and argue for their correctness. Most of the groundwork has already been established in \cref{sec:proof_correctness}; to avoid repeating too much material, we will proceed much more informally than before, primarily emphasizing differences from \cref{sec:proof_correctness} and the changes that must be made to adapt to the General Multiple-Squares HSP.

\subsection{The algorithm and its subroutines}

The algorithm proceeds almost identically, with only the details of the subroutines modified.

\begin{widetext}
General Multiple-Squares StateHSP algorithm:
\begin{enumerate}
    \item Perform the \textbf{Learning Pauli-like Stabilizers} subroutine (\cref{def:pauli_sampling}) once, which simply performs Fourier sampling on the $\mathbb{Z}_2^N \times \mathbb{Z}_2^N$ subgroup of $\mathcal{D}_4^N$ comprising of elements where $w = 0$ everywhere. If the $w^h = 0$ everywhere or $t^h = 0$ everywhere, we are finished. Otherwise, continue.
    \item Determine $w^\text{max}$, which gives the conjugacy class of the reflections on each factor in the hidden element~$h$.
    This is done in two steps:
    \begin{enumerate}
        \item Perform Fourier sampling on the $\mathbb{Z}_2^N\times\mathbb{Z}_2^N$ quotient of $\mathcal{D}_4^N$ by its center to generate the annihilator subgroup $K^\perp$.
        Samples $(p,q)$ are obtained by
        \begin{enumerate}
            \item Repeating the \textbf{partial Fourier sampling subroutine} (\cref{def:partial_fourier_sampling}) until one accumulates a \textbf{Bell-resolvable set} that includes the first sample (requiring independent sampling of the first parity outcome ensures sufficiently rich Bell-resolvable sets).
            \item Performing the \textbf{Bell-resolution subroutine}. This is identical to before, but with a modified classical processing step (\cref{sssec:mod_bell_res}).
        \end{enumerate}
        \item Obtain $w^\text{max}$ from $K^\perp$ using the \textbf{maximal rotation subroutine}.
    \end{enumerate}
    \item Use $w^\text{max}$ to determine where to apply the \textbf{S correction} (\cref{def:S_correction}) to $|\Psi\rangle$. This converts $|\Psi\rangle$ to something stabilized by a $R^N(t^h, v^h, 0)$.
    \item Learn $t^h, v^h$ with \textbf{Learning Pauli-like Stabilizers}. 
    \item Compute $w^h = t^h \odot w^{\rm max}$ to obtain $w^h$. 
\end{enumerate}
\end{widetext}

\subsubsection{The Modified Subroutines}

\begin{defn}[Learning Pauli-like Stabilizers]\label{def:pauli_sampling}
  The Learning Pauli-like Stabilizers subroutine takes as input a state $| \Psi \rangle$ fulfilling the General Multiple-Squares StateHSP promise, and outputs any hidden subgroup of the form $\{(0, 0, 0), (t, v, 0)\}$ if it exists. Otherwise, it gives a null output (or equivalently the trivial subgroup).
  It consists of the following steps:
  \begin{enumerate}
      \item On $O(N/\varepsilon)$ copies of $| \Psi \rangle$, act with $U_{\rm pFT}^{\dag}$ on every site, then perform $X$ measurements on the $D$ register. 
      \item Process the output as an abelian StateHSP to identify any hidden subgroup.
  \end{enumerate}
\end{defn}

This performs Fourier sampling on the abelian subgroup
\begin{equation}
    \{(t, v, 0) \in \mathcal{D}_4^N : \forall t, v \in \mathbb{Z}_2^N \} \equiv \mathbb{Z}_2^N \times \mathbb{Z}_2^N \leq \mathcal{D}_4^N,
\end{equation}
which is a representation of the phaseless Pauli group $\mathcal{P}_N/\{\pm 1, \pm i\}$.

The correctness of this subroutine is quickly shown, by noting that \cref{def:pauli_sampling} performs Fourier sampling on this subgroup, which follows from \cref{eq:pauli_irreps} and the surrounding discussion. Hence, as a corollary to \cref{fact:abelian_stateHSP}, we obtain
\begin{cor}[Learning Pauli-like Stabilizers]
    Suppose the hidden stabilizer is of the form $(t^h, v^h, 0)$. Then the procedure of \cref{def:pauli_sampling} learns $t^h, v^h$ with high probability using $O(N/\varepsilon)$ samples.
\end{cor} 

The next subroutine also involves Fourier sampling, but on the full $\mathcal{D}_4^N$ instead. 

\begin{defn}[Partial Fourier Sampling]\label{def:partial_fourier_sampling}
    The Partial Fourier Sampling subroutine takes as input a state $| \Psi \rangle$ fulfilling the Regular Multiple-Squares StateHSP promise, and gives as output a triplet of bitstrings $\alpha, \sigma, \pi \in \mathbb{Z}_2^N$, and a state $| \Psi(\sigma, \pi) \rangle$, via the following steps:
    \begin{enumerate}
        \item Perform $U_{\rm pFT}^{\dag}$ on the state $| \Psi \rangle$.
        \item Measure the first three registers in the computational basis to obtain the rep labels $(\alpha, \sigma, \pi)$. Note that $\alpha$ is never used.
        \item The remaining post-measurement state is $| \Psi^{\alpha}(\sigma, \pi) \rangle$, and is an $N$ \textit{qubit} state.
    \end{enumerate}
\end{defn}

To account for the larger number of new outputs from partial Fourier sampling, we provide a modified (almost identical to before) definition of the Bell-resolvable set.

\begin{defn}[Bell-resolvable set II]\label{def:Bell_resolvable_set_II}
  A Bell-resolvable set is a set of bitstrings $\{\alpha^m, \sigma^m, \pi^m \in \mathbb{Z}_2^N\}_{m=0}^M$, obtained from the partial Fourier sampling subroutine, which satisfy $\sum_{m=0}^M \pi^m_j = 0$, together with the state $\bigotimes_{m=0}^M | \Psi^{\alpha^m}(\sigma^m, \pi^m) \rangle$.
\end{defn}

The Bell-resolution subroutine is identical to before, up to a reinterpretation of measurement results, which we will explicitly spell out in \cref{sec:processing-classical-outcomes} for completeness. Similarly, the maximal rotation subroutine is exactly identical to before. 

The final subroutine is a generalization of applying $T$ gates in the previous algorithm.
\begin{defn}[S Correction]\label{def:S_correction}
    Given $w^{\rm max} \in \mathbb{Z}_2^N$, the S correction refers to the following circuit: 
    \begin{equation}
        \bigotimes_{j=1}^N (U_{\rm pFT})_j cc_{S, P \rightarrow D}Z_j (c_{P \rightarrow D} S^{\dag})^{w^{\rm max}_j} (U_{\rm pFT})_j^{\dag},
    \end{equation}
    where $S = e^{i \pi Z/4}$.
\end{defn}

This corresponds to the $T^{\dag} \otimes T^{\dag}$ part of the previous algorithm. Note that while we have written the above as a coherent controlled operation, we may push this through the final Learning Pauli-like Stabilizers, on which we only require classically controlled gates. We will briefly describe this later.

Having defined the modified subroutines, what do we have to show? In \cref{ssec:pFT_anticoncentration}, we will show an analogue of \cref{prop:basis_sampling_complexity}, which demonstrates that partial Fourier Sampling shares the same anti-concentration properties as parity sampling. In \cref{sssec:mod_bell_res}, we spell out the modified conjugacy ensemble representation as a representation of $\mathbb{Z}_2^N \times \mathbb{Z}_2^N$. In \cref{sssec:good_nullspaces_still_possible}, we will argue that good null spaces are still possible in this more general setup. Finally, in \cref{sssec:S_corr_correct}, we will show that the S correction does map the hidden element to an abelian subgroup of $\mathcal{D}_4^N$, leading to an abelian StateHSP.
Together, these constitute a proof for \cref{thm:Dihedral_StateHSP_theorem_II}.

\subsection{The generalized subroutines}

\subsubsection{Partial Fourier Sampling is anti-concentrated}\label{ssec:pFT_anticoncentration}

In this section, we will argue that partial Fourier sampling is anti-concentrated in the parity $\pi$. First, we want to understand what measuring the parity does in terms of a group action. 

\begin{lem}[Parity sampling measures a group action]\label{lem:parity_sampling_group_action}
    Let $(\sigma_j, \pi_j)$ be some output of the Partial Fourier Sampling subroutine on the $j$-th site. Then, $\pi_j$ is distributed as the measurement outcomes of the operator $R(0, 1, 0)$.
\end{lem}

\begin{proof}
    The partial Fourier Sampling subroutine takes $| \Psi \rangle$, acts on it with $U_{\rm pFT}^{\dag}$, and does a $Z$ basis measurement on the $\pi$ register. Hence, we just need to compute $U_{\rm pFT} Z_{P} U_{\rm pFT}^{\dag}$ to get the associated observable, where $Z_P$ simply denotes a $Z$ on the $P$ register, which may be represented as, 
    \begin{equation}
        \begin{aligned}
            &Z_{P} \\  
            &= \sum_{\sigma \in \mathbb{Z}_2} \sum_{\alpha = 0}^{M_{\sigma, 0}} | \alpha \rangle \langle \alpha |_A \otimes | \sigma \rangle \langle \sigma |_S \otimes | 0 \rangle \langle 0 |_P \otimes I \\
            &\qquad -  \sum_{\sigma \in \mathbb{Z}_2} \sum_{\alpha = 0}^{M_{\sigma,  1}} | \alpha \rangle \langle \alpha |_A \otimes | \sigma \rangle \langle \sigma |_S \otimes  | 1 \rangle \langle 1 |_P \otimes I\\
            &= \sum_{\sigma, \pi \in \mathbb{Z}_2} \sum_{\alpha = 0}^{M_{\sigma, \pi}} | \alpha \rangle \langle \alpha | \otimes | \sigma \rangle \langle \sigma | \otimes | \pi \rangle \langle \pi | \otimes (-1)^{\pi} I \\
            &= \sum_{\sigma, \pi \in \mathbb{Z}_2} \sum_{\alpha = 0}^{M_{\sigma, \pi}} | \alpha \rangle \langle \alpha | \otimes | \sigma \rangle \langle \sigma | \otimes | \pi \rangle \langle \pi | \otimes D^{2 \sigma + \pi}(0, 1, 0).
        \end{aligned}
    \end{equation}
    Hence, applying the inverse of \cref{eq:complete_pFT}, we have 
    \begin{equation}
        U_{\rm pFT} Z_{\pi} U_{\rm pFT}^{\dag} = R(0, 1, 0).
    \end{equation}
\end{proof}

Using \cref{lem:parity_sampling_group_action}, we can write the outcome of measurements of $\pi \in \mathbb{Z}_2^N$ across all sites as a sum of expectation values of representations of the form $R(0, v, 0)$. Following the proof of \cref{prop:basis_sampling_complexity}, we obtain:

\begin{prop}[Parity sampling complexity in general representations ]\label{prop:parity_sampling_complexity_general}
    Let $|\Psi \rangle$ be a state satisfying the assumptions of the General Multiple-Squares StateHSP. Suppose we have, in Step 1 of the algorithm, ruled out the possibility that $t^h = 0$, and let us obtain bitstrings $\pi^m$ from $| \Psi \rangle$ via the partial parity sampling subroutine. To obtain a set of $\{\pi^m\}_{m=1}^M$ (along with arbitrary $\alpha^m, \sigma^m$) that span $\mathbb{Z}_2^N$ with high probability, it suffices to obtain $O(N / \varepsilon)$ samples. 
\end{prop}

This tells us that we can consistently and efficiently form Bell-resolvable sets in Step 3 of the algorithm.

\subsubsection{Modified Bell-resolution subroutine}\label{sssec:mod_bell_res}

First, we should characterize the analogue of the conjugacy ensemble representation. In particular, the Bell-resolvable set now involves two sets of bitstrings, $\{\sigma^m, \pi^m\}_{m=0}^M$, where we only require $\sum_{m=0}^M \pi^m = 0$. We can form a similar conjugacy ensemble representation from Bell-resolvable sets;
\begin{equation}
    \begin{aligned}
        \bigotimes_{m=0}^M &D^{2 \sigma^m + \pi^m}(t, v, w)
        =
        D_{\mathcal{C}}^{\{\sigma^m, \pi^m \}}(t, w)
        \\
        &=
        (-1)^{\sum_j w_j (\sum_m \sigma^m_j)} \bigotimes_{m=0}^M D^{\pi^m}_{\mathcal{C}}(t, w)
        \,,
    \end{aligned}
\end{equation}
where the dependence on $v$ disappears due to the Bell-resolvable requirement. We see that the only difference from the conjugacy ensemble representation is the multiplication by a sign that depends on $w$, which respects addition in $w$. Thus, direct computation shows that: 
\begin{prop}[Conjugacy ensemble representation]\label{prop:conjugacy-ensemble-representation}
  A Bell-resolvable set (\cref{def:Bell_resolvable_set_II}) specifies a conjugacy ensemble representation, 
  \begin{equation}
      D_{\mathcal{C}}^{\{\sigma^m, \pi^m\}}(t, w) = (-1)^{ \sum_j (\sum_m \sigma_j^m) w_j}\bigotimes_{m=0}^M D_{\mathcal{C}}^{\pi}(t, w),
  \end{equation}
  which satisfies the properties:
  \begin{enumerate}
      \item $D_{\mathcal{C}}^{\{\sigma^m, \pi^m\}}$ is a representation of $\mathbb{Z}_2^N \times \mathbb{Z}_2^N$.
      \item $D_{\mathcal{C}}^{\{\sigma^m, \pi^m\}}(t^h, w^h)$ stabilizes the Bell-resolvable state.
  \end{enumerate}
\end{prop}

\subsubsection{Good null spaces are still possible}\label{sssec:good_nullspaces_still_possible}

Having proven \cref{lem:parity_sampling_group_action},  the fact that good nullspaces are still possible follows quite naturally, and almost identically from the way \cref{prop:bad_null_spaces_unlikely} was proven. 

To understand why, we first use \cref{lem:parity_sampling_group_action} to show that the General Multiple-Squares StateHSP promise, together with the guarantee from Step 1 that $t^h \neq 0$, implies that the probability mass of maximum proper subgroup of $\pi$'s is bounded away from $1$ by a constant $\varepsilon$.
Next, we note that $D_{\mathcal{C}}^{2 \sigma + \pi}$ only differs from $D_{\mathcal{C}}^{\pi}$ by a sign, so their commutation relations, as exemplified by \cref{lem:repn_cmm_form} is unaffected.
This means that analogues to \cref{lem:bad_elt_anticommute}, and \cref{lem:bad_elt_commute} hold.
This allows us to prove a direct analogue to \cref{prop:bad_null_spaces_unlikely}, which shows that we can exclude bad null spaces by sampling $O(N/\varepsilon)$ Bell-resolvable sets.

\subsubsection{Processing of classical outcomes}
\label{sec:processing-classical-outcomes}

Fourier sampling over $\mathbb{Z}_2^N\times\mathbb{Z}_2^N$ gives us pairs of vectors $(q,p)$ corresponding to the irreps.
Each component of the vectors $q$ and $p$ corresponds to one of the factors of the original $\mathcal{D}_4^N$ direct product.
When we obtain a Bell-resolvable set, each of these factors is associated with a tensor product of representations of $\mathbb{Z}_2\times\mathbb{Z}_2$ taken from the representations given in \cref{eq:irrep-bases} using the notation for the irreps of $\mathbb{Z}_2\times\mathbb{Z}_2$ from \cref{eq:z2z2-irreps} and the Bell states defined in \cref{eq:bell-state-defn}.

We do Fourier sampling on these representations by either making a Bell measurement (for the two-qubit representations) or measuring in the $X$ basis (for the single-qubit representations).
For each of the $N$ sites in the direct product, we collect all the $(q,p)$ pairs from projecting onto the $C^{q,p}$ irrep and combine them by separately summing over the $q$s and $p$s modulo 2.
This gives us the label for the irrep obtained by the tensor product of those irreps on that factor, since $C^{q_1,p_1}C^{q_2,p_2}=C^{(q_1+q_2,p_1+p_2)}$.
The final character for $\mathbb{Z}_2^N\times\mathbb{Z}_2^N$ obtained by the Fourier sampling procedure is the pair of vectors $(q,p)$ obtained by performing this sum for each factor.

\subsubsection{The S correction}\label{sssec:S_corr_correct}

Having shown that good null spaces are still possible, and having defined the S correction in \cref{def:S_correction}, we just have to show that it maps the hidden element into the center of the group. This will allow us to use the Center Sampling subroutine, \cref{def:pauli_sampling}, to identify $(t^h, v^h)$.

\begin{prop}[S Correction to the Paulis]\label{prop:S_correction_center}
    Given $w^{\rm max} \in \mathbb{Z}_2^N$ such that $t^h \odot w^{\rm max} = w^h$, the S correction, \cref{def:S_correction}, maps $R^N(t^h, v^h, w^h) \mapsto R^N(t^h, v^h, 0)$.
\end{prop}

\begin{proof}
    Recall that the S correction is the unitary 
    \begin{equation}
        U_{T} = \bigotimes_{j=1}^N (U_{\rm pFT})_j cc_{S, P \rightarrow D} Z_j (c_{P \rightarrow D}S_j)^{w^{\rm max}_j} (U_{\rm pFT})_j^{\dag}.
    \end{equation}
    We need to work out $U_{cT} R^N(t^h, v^h, w^h) U_{cT}^{\dag}$. First, the $U_{\rm pFT}$ part simply block-diagonalizes $R$ so that we can work with the $D$ reps. The $cc_{S, P \rightarrow D}Z_j$ part only acts on $D^3$, and turns it into $D^1$. After these two transformations, $R^N(t^h, v^h, w^h)$ gets mapped to,
    \begin{equation}
    \begin{aligned}
        &\bigotimes_{j=1}^N \sum_{\sigma_j, \pi_j \in \mathbb{Z}_2} \sum_{\alpha_j = 0}^{M_{\sigma, \pi}} \\
        & | \alpha_j, \sigma_j, \pi_j \rangle \langle \alpha_j, \sigma_j, \pi_j |_{A, S, P} \otimes D^{2 \delta_{\pi, 0} \sigma + \pi}(t^h_j, v^h_j, w^h_j).
    \end{aligned}
    \end{equation}
    We can split the calculation of the effect of the $cS_{P \rightarrow D}$ into parts, first by $\pi$, then by $t^h_j = 1, 0$. First, for $\pi_j = 0$, the action is trivial, so nothing happens. 
    
    Next, for $\pi_j = 1$, we will apply an $S$ gate. Recall that $S=e^{i\pi Z/4}$. When $t^h_j = 0$, we must have $w^h_j = 0$ (since the hidden element is order $2$). Thus $D^{1}(0, v^h_j, 0)$ proportionate to the identity; the $T$ conjugation has no effect on it. 
    
    Finally, we consider the case where $t^h_j = 1, \pi_j = 1$. For such $j$, we have $w^{\rm max}_j = w^h_j$. Omitting the $A, S, P$ registers, we see that the action of the $c_{P \rightarrow D} S$ gate is 
    \begin{equation}
        \begin{aligned}
            (S^{\dag})^{w^h_j} D^1(1, v^h_j, w^h_j) S^{w^h_j} &= D^1(1, v^h_j, w^h_j) (S^2)^{w^h_j} \\
            &= D^1(1, v^h_j, w^h_j) D^1(0, 0, w^h_j)^{\dag}\\
            &= D^1(1, v^h_j, 0).
        \end{aligned}
    \end{equation}

    Finally, to put it all together, we note that already $D^0(t^h_j, v^h_j, w^h_j) = D^0(t^h_j, v^h_j, 0)$ and $D^2(t^h_j, v^h_j, w^h_j) = D^2(t^h_j, v^h_j, 0)$, as these representations don't contain any information about parity. Hence, just after the controlled-S step, the conjugation gives 
    \begin{equation}
    \begin{aligned}
        &\bigotimes_{j=1}^N \sum_{\sigma_j, \pi_j \in \mathbb{Z}_2} \sum_{\alpha_j = 0}^{M_{\sigma, \pi}} \\
        &| \alpha_j, \sigma_j, \pi_j \rangle \langle \alpha_j, \sigma_j, \pi_j |_{A, S, P}  \otimes D^{2 \delta_{\pi, 0} \sigma + \pi}(t^h_j, v^h_j,0).
    \end{aligned}
    \end{equation}
    Acting on this with $U_{\rm pFT}$ yields $R(t^h_j, v^h_j, 0)$, and the result follows.
\end{proof}

Finally, the above can be used to prove an analog to \cref{prop:maximal_T_rotation_abelian_stateHSP}, so that subsequent to the $S$ correction, we may solve for $t^h, v^h$ using the Learning Pauli-like Stabilizers subroutine (\cref{def:pauli_sampling}), as claimed.

\subsubsection{S correction via classical control}

Finally, in the above definition, we have written the $S$ correction as a quantum-controlled operation. This allows one to produce a state that exactly fulfills the abelian StateHSP promise. However, we note that  we may actually push this $S$ correction through the last Learning Pauli-like Stabilizers step, which removes the need for coherent control via the irrep labels. The reason for this is that the Fourier sampling for the subsequent Learning Pauli-like Stabilizers can be carried out by first measuring the $A, S, P$ registers, before carrying out a single qubit measurement on the $D$ register. As such, since we are always projecting onto a particular computational basis state of the $A, S, P$ registers, there is no need to make sure the $S$ correction is coherent over these registers, instead, we can measure these registers, and then based on the measurement result (as well as the value of $w^\text{max}$) opt whether or not to perform the $S, Z$ gates, before carrying out the final single qubit measurement. This will still yield the outcome distribution of Fourier sampling on the claimed abelian StateHSP.

\section{Conclusion}

In this work, we have provided a direct solution to the StateHSP for a non-abelian group, and pointed out connections between the StateHSP and Hamiltonian spectroscopy. This provides an example of how one may adapt the weak Fourier sampling approach of previous works to a simple non-abelian group. Looking forward, there are many classes of non-abelian groups for which we can study the StateHSP. One takeaway from this work is that the irrep structure of the representation is crucial to solving non-abelian StateHSPs, hence one interesting future direction is to identify non-abelian StateHSPs which admit representation-specific solutions, even when the relevant non-abelian HSP does not admit an efficient solution.

\begin{acknowledgments}
We gratefully acknowledge useful conversations with Dave Bacon, Sergio Boixo, Ryan Babbush, Vadim Smelyanskiy, and James Watson, and thank Xinyu Liu, Robert Huang, and Robbie King for pointing out the application of our algorithm to the StateHSP.
\end{acknowledgments}

\bibliography{apssamp}

\end{document}